\documentclass{acm_proc_article-sp}

\usepackage{amsmath,amssymb}
\usepackage{graphicx,color}
\usepackage{url}

\newtheorem{theorem}{Theorem}

\newtheorem{example}[theorem]{Example}

\newtheorem{lemma}[theorem]{Lemma}
\newtheorem{proposition}[theorem]{Proposition}

\newtheorem{definition}[theorem]{Definition}

\newenvironment{remark}{\textbf{Remark}}

\newcommand{\<}{\langle}
\newcommand{\rg}{\rangle}
\newcommand{\ud}{\mathrm{d}}
\newcommand{\oomit}[1]{}
\newcommand{\define}{\widehat{=}}
\newcommand{\uu}{\mathbf{u}}

\newcommand{\xx}{\mathbf{x}}

\newcommand{\fb}{\mathbf{f}}
\newcommand{\QQ}{\mathbb{Q}}
\newcommand{\RR}{\mathbb{R}}

\def \la {\langle}
\def \ra {\rangle}

\begin{document}



\title{Computing Semi-algebraic Invariants for Polynomial Dynamical Systems 
}

\numberofauthors{3}

\author{
\alignauthor Jiang Liu  \\
 \affaddr{State Key Lab. of Comp. Sci. \\
 Institute of Software \\
 Chinese Academy of Sciences} \\
  \email{liuj@ios.ac.cn}
 \alignauthor
Naijun Zhan\\
   \affaddr{State Key Lab. of Comp. Sci. \\
    Institute of Software \\
   Chinese Academy of Sciences}\\
  \email{znj@ios.ac.cn}
 \alignauthor
Hengjun Zhao
\\
   \affaddr{State Key Lab. of Comp. Sci. \\
   Institute of Software \\
    Chinese Academy of Sciences} \\
  \email{zhaohj@ios.ac.cn}
}

\maketitle

\begin{abstract}
In this paper, we consider an extended concept of invariant for polynomial dynamical system (PDS) with domain and initial condition, and establish a sound and complete criterion for checking semi-algebraic invariants (SAI) for such PDSs. The main idea is encoding relevant dynamical properties as conditions on the high order Lie derivatives of polynomials occurring in the SAI. A direct consequence of this criterion is a relatively complete method of SAI generation based on template assumption and semi-algebraic constraint solving. Relative completeness means if there is an SAI in the form of a predefined template, then our method can indeed find one using this template.

\end{abstract}

%


\section{Introduction}
Hybrid systems are those systems involving both continuous evolutions and discrete transitions. How to design correct (desired) hybrid systems is a grand challenge in computer science and control theory.
From a computer scientist's point of view, the main concern on hybrid systems up to now is to verify so-called safety properties. A safety property claims that some unsafe state is never reachable from any initial state along with
any trajectory of the system.

\subsection{Motivation}
Directly computing the reachable set is a natural way to address this issue. As we know, there are two well-developed techniques for computing reachable set so far, that is, techniques based on
model-checking \cite{CE81,SQ82} and the decision procedure of Tarski algebra \cite{tarski51}, respectively. However, the former technique requires the decidability and therefore can only be applied to some simple hybrid systems, e.g.
timed automata \cite{Alur94}, multirate automata \cite{Alur95}, rectangular automata \cite{Puri94,Henzinger95}, and so on. Comparably speaking, the latter technique has a wider scope of applications. For example,  in \cite{LPY02}  how to compute reachable sets for three classes of special linear hybrid systems are investigated. However, this technique heavily depends on whether  the explicit solutions of the considered differential equations are or can be reduced to polynomials.  So, this approach can not be applied to general linear hybrid systems, let alone  nonlinear systems.

To deal with more complicated systems, recently, a deductive method has been established and successfully applied in practice \cite{PlatzerClarke08,PlatzerClarke09}, which can be seen as a generalization of the so-called Floyd-Hoare-Naur inductive assertion method.
Inductive assertion method is thought to be the dominant method for the verification of sequential programs. To generalize the inductive
method to hybrid systems, a logic similar to Hoare logic which can deal with continuous dynamics is necessary. For example, differential-algebraic dynamic logic \cite{Platzer10} due to Platzer was invented by extending dynamic logic with continuous statements. Recently, Liu et al \cite{aplas} had another effort by extending Hoare logic to hybrid systems for the same purpose.

The most challenging part of the inductive method is how to discover invariants of hybrid systems.  An invariant is a property that holds at all reachable states from any initial state that satisfies this property. If we can get invariants that are strong enough to imply the safety property to be verified, then we succeed in safety verification without solving differential equations, while differential equations have to be exactly solved or approximated in the methods via directly
computing reachable sets. In particular, if the term expressions of a hybrid system are or can be reduced to polynomials, the so-called \emph{inductive invariants} \cite{Manna04} can be effectively generated using the constraint-based approach \cite{Tiwari08}.

The key issue in generating inductive invariants of a hybrid system is to deal with continuous dynamics, i.e. to generate {continuous invariant} of the continuous evolution at each location (mode) of the hybrid system. A location (mode) of a hybrid system is usually represented by a \emph{ continuous dynamical system with domain and initial condition} (CDSwDI for short) of the form $(H,\mathbf{f},\Xi)$, where $\mathbf{f}$ is a vector field, $H$ is a domain restriction of continuous evolution, and $\Xi\subseteq H$ is a set of initial states.
A property $\varphi$ is called a \emph{continuous invariant} (CI for short) of $(H, \mathbf{f},\Xi)$, if
it is always satisfied along any trajectory whose starting point satisfies $\Xi$, as long as the trajectory still remains in domain $H$. For $\varphi$ to be a CI of $(H, \mathbf{f},\Xi)$, the more complex the forms of $H$, $\fb$, $\Xi$ and $\varphi$ are, the more intricate constraints should be induced accordingly.
A global (discrete) inductive invariant of a hybrid system consists of a set of CIs such that: the initial condition of the initial location  (mode) entails the CI of the initial location, and if there is a discrete transition between two locations of the system,
then the CI at the pre-location implies the CI at the post-location w.r.t. the discrete transition.
There are many methods, e.g. \cite{yzzx10}, for certifying and generating global inductive invariants of a system by using the global inductiveness. Therefore in this paper we only focus on how to generate CI at a single location (mode), i.e. a CDSwDI.

\subsection{Related Work}
In the literature, lots of efforts have been made towards algebraic or semi-algebraic continuous invariants generation for polynomial dynamical systems, even though CI may have different synonyms.

The generation of algebraic invariants, i.e. sets defined by polynomial equations are usually based on the theory of ideals in polynomial ring.
In \cite{Manna04}, to handle continuous differential equations, two strong continuous consecution conditions are imposed on the predefined templates, and then the two conditions are encoded as ideal membership statements.
The work in \cite{Tiwari05} showed that the set of algebraic invariants of a linear system, which forms a polynomial ideal, is computable. The above two approaches both use \emph{Gr\"obner bases} computation. An efficient technique that computes algebraic invariants as the greatest fixed point of a monotone operator over \emph{pseudo ideals} was presented in \cite{Sankar10}.

As for the polynomial inequality case, to guarantee that $p\geq 0$ is a CI of a PDS $(H, \mathbf{f},\Xi)$, it is useful to analyze the direction of $\fb$  with regard to the set $p\geq 0$.
In \cite{Prajna04,Prajna07}, the authors proposed the notion \emph{barrier certificates} for safety verification of hybrid systems. A polynomial $p$ could be a barrier certificate if the unsafe region is included in $p<0$, and at any point in $p=0$, $\fb$ points (strictly) inwards the set $p\geq 0$.
Such polynomial barrier certificates can be effectively computed using sum of squares decomposition and semi-definite programming. In \cite{Tiwari08} a similar idea is adopted and by reducing the conditions of CI to semi-algebraic constraints, invariants that are boolean combinations of polynomial equations and inequalities can be generated. Unfortunately, the
approaches in \cite{Prajna04,Tiwari08} were discovered  in \cite{Taly09,Tiwari10,Platzer10} to have certain problems with their soundness, if at the boundary of a CI, $\fb$ is not strictly inward the invariant set.
In \cite{PlatzerClarke08} the authors proposed the notion of \emph{differential invariant} and the principle of differential induction. Basically, $p\geq 0$ is a differential invariant of $(H, \mathbf{f},\Xi)$ if at any point in $H$, the directional derivative of $p$ in the direction of $\fb$ is non-negative. Such requirement is strong, but provide a sound and effective way of generating complex semi-algebraic continuous invariants.

\subsection{Our Contribution}
The problem of checking inductiveness for continuous dynamical systems was considered in \cite{Taly09} and \cite{Tiwari10}. Therein various sound checking rules are presented, which are also complete for classes of continuous invariants, e.g. linear, quadratic, convex and smooth invariants. The authors even proposed a sound and relatively complete rule using higher order \emph{Lie derivatives}, which is quite similar to ours. However, in their relatively complete rule there are infinitely many candidate tests and thus is computationally infeasible. Our work in this paper actually resolves this problem and completes the gap left open in \cite{Taly09, Tiwari10}

The relative completeness of our method means that for a given PDS, if there is an SAI  of the predefined
template, then our method can indeed discover one SAI using this template.
Thus, there are two advantages with our approach comparing with the well-established approaches: firstly, more general SAIs can be generated; secondly, a by-product of the completeness of our approach is that whether a given semi-algebraic set is really an SAI of a given PDS is decidable. This is quite useful in the interplay of discrete invariant generation (global) and CI generation (local).

\subsection{Paper Organization}
The rest of this paper is organized as follows.  Section \ref{sec:prelim} presents some basic notions and fundamental theories  on  algebraic geometry and dynamical system. Section \ref{sec:di} gives a formal definition of the  SAI generation problem. In Section \ref{sec:fund}, we prove the fundamental results based on which our method is developed. Section \ref{sec:idea} illustrates the basic idea of our approach in simple cases. How to apply our approach to general cases is investigated  in Section \ref{sec:AIG}. Two case studies are given in Section \ref{sec:case}. Section \ref{sec:con} concludes this paper with a discussion of future work.

\section{Preliminaries}\label{sec:prelim}

In this section, we will recall some basic notions.

\subsection{Polynomial Ideal Theory}
Let $\mathbb{K}$ be an algebraic field and $\mathbb{K}[x_1, \dots, x_n]$ denote the polynomial ring with
coefficients in $\mathbb{K}$. In this paper, $\mathbb{K}$ will be taken as the rational number field $\QQ$. Customarily, let $\xx$ denote
the $n$-tuple $(x_1, \cdots, x_n)$ with $\dim{(\xx)}=n$, and a polynomial in $\mathbb{Q}[x_1, \dots, x_n]$ ($\mathbb{Q}[\xx]$ for short) may be written as $p(\xx)$ or $p$ simply. A parametric polynomial
\[p(\uu,\xx)\in \QQ[u_1, u_2, \ldots, u_t, x_1, x_2, \ldots,
x_n]\] is called a {\it template}, where $\xx$ are variables taking values from $\mathbb{R}^n$ and $\uu$ are coefficient parameters taking values from $\mathbb{R}^t$. Given $\uu_0\in \mathbb{R}^t$, we call the polynomial $p_{\uu_0}(\xx)$ resulted by substituting
$\uu_0$ for $\uu$ in $p(\uu,\xx)$ an {\it instantiation} of $p(\uu,\xx)$.

In what follows, we recall the theory of polynomial ideal (refer to \cite{clo}).
\begin{definition}
A subset $I\subseteq \mathbb{K}[\xx]$ is called an \emph{ideal} if
\begin{enumerate}
\item[i)] $0\in I$.
\item[ii)] If $p(\xx), g(\xx)\in I$, then $p(\xx)+g(\xx)\in I$.
\item[iii)] If $p(\xx)\in I$ and $h(\xx)\in \mathbb{K}[\xx]$, then $p(\xx)h(\xx)\in I$.
\end{enumerate}
\end{definition}
It is easy to check that if $p_1, \cdots, p_k\in \mathbb{K}[\xx]$, then
\[\<p_1, \cdots, p_k\rg = \{\sum_{i=1}^{k}p_i h_i\mid \forall i \in [1,k]. \,h_i \in
\mathbb{K}[\xx]\}\]
is an ideal. In general, we say an ideal $I$ is {\itshape generated} by polynomials $g_{1}, \dots, g_k\in \mathbb{K}[\xx]$ if $I=\<g_{1},   \dots, g_k\rg$, and $\{g_1, \ldots, g_k\}$  is called a set of \emph{generators}  of $I$.
\begin{theorem}[Hilbert Basis Theorem]
Every ideal \, $I\subseteq \mathbb{K}[\xx]$ has a finite generating set. That is, $I=\<g_{1}, \dots, g_k\rg$ for some $g_{1}, \dots, g_k
\in \mathbb{K}[\xx]$.
\end{theorem}
For its proof, please refer to \cite{clo}. Based upon this result, it is easy to see that
\begin{theorem}[Ascending Chain Condition]\label{ACC}
For any ascending chain
\[I_1\subseteq I_2 \subseteq \cdots \subseteq I_\ell \subseteq \cdots\] of ideals in polynomial ring $\mathbb{K}[\xx]$, there must be $N$ such that for all $\ell\geq N$, $I_\ell=I_N$.
\end{theorem}

\subsection{Semi-algebraic Set}\label{subsec:semi-alg}

An atomic polynomial formula over variables $x_1,x_2,\ldots,x_n$ is $p\,\triangleright\,0$, where
$p$ is a polynomial in $\QQ[\xx]$ and $\triangleright\in \{\geq,>,\leq,<,=,\neq \}$. A quantifier free polynomial formula is a boolean combination of atomic polynomial formulas using connectives $\vee,\wedge,\neg,\rightarrow,$ etc.
\begin{definition}[Semi-algebraic Set]\label{dfn:sas}
A subset $S$ of $\RR^n$ is called a semi-algebraic set, if there is a quantifier free polynomial formula $\varphi$ s.t.
$$S=\{\xx \in \RR^n\mid \varphi(\xx) \,\textrm{ is true}\}\enspace.$$
\end{definition}
We will use the  $\mathcal S(\varphi)$ to denote the semi-algebraic set defined by a quantifier free polynomial formula $\varphi$. It is easy to check that any semi-algebraic set can be transformed into the form
$$\mathcal S (\bigvee_{i=1}^{I}\bigwedge_{j=1}^{J_i}p_{ij}\triangleright 0),\,\mbox{where }\, \triangleright \in \{\geq,>\}\enspace .$$

Note that semi-algebraic sets are closed under basic set operations, since
\begin{itemize}
\item $\mathcal S(\varphi_1)\cap\mathcal S(\varphi_2)=\mathcal S(\varphi_1\wedge \varphi_2)$\,;
\item $\mathcal S(\varphi_1)\cup\mathcal S(\varphi_2)=\mathcal S(\varphi_1\vee \varphi_2)$\,;
\item $\mathcal S(\varphi_1)^c=\mathcal S(\neg \varphi_1)$\,;
\item $\mathcal S(\varphi_1)\setminus\mathcal S(\varphi_2)=\mathcal S(\varphi_1)\cap\mathcal S(\varphi_2)^c=\mathcal S(\varphi_1 \wedge \neg\varphi_2)$\,,
\end{itemize}
where $A^c$ and $A\setminus B$ stand for the complement and subtraction operation of sets respectively.

\subsection{Continuous Dynamical System}\label{subsec:Lie}
We recall the theory of
continuous dynamical systems in the following. Please refer to \cite{Haddad-Che} for details.

\subsubsection{Trajectories of Continuous Dynamical System}
An autonomous \emph{continuous dynamical system} (CDS) is modeled by first-order ordinary differential equations
\begin{equation}\label{eq:ode}
  \dot \xx= \fb(\xx) \enspace,
\end{equation}
where $\xx\in\mathbb{R}^n$ and $\fb$ is a vector function from $\mathbb R^n$ to $\mathbb R^n$, which is also called a \emph{vector field} in $\mathbb R^n$.

If $\fb$ satisfies the \emph{local Lipschitz condition}, then given
$\xx_0\in \mathbb R^n$, there exists a unique solution $\xx(\xx_0;t)$ of
(\ref{eq:ode}) defined on $(a, b)$ with $a< 0< b$ s.t.
$$\forall t\in (a,b).\,{\ud \xx(\xx_0,t)\over \ud t} = \fb (\xx(\xx_0;t))\quad \mathrm{and}\quad  \xx(\xx_0;0)=\xx_0.$$
When $\xx_0$ is clear from the context, we just write $\xx(\xx_0;t)$ as $\xx(t)$. Based upon this, we shall use the
following useful notions for our discussion in the sequel.
\begin{definition}[Trajectory] Suppose $\xx(\xx_0;t)$ is the solution to (\ref{eq:ode}) defined on $(a,b)$ with $a< 0< b$, as stated above. Then
  \begin{itemize}
    \item $\xx(\xx_0;t)$ ($\xx(t)$ for short) defined on $[0,b)$ is called the trajectory of (\ref{eq:ode}) starting from $\xx_0$;
    \item $\xx(\xx_0;-t)$ ($\xx(-t)$ for short) defined on $[0,-a)$, resulted by substituting $-t$ for $t$ in $\xx(\xx_0;t)$, is called the inverse trajectory of (\ref{eq:ode}) starting from  $\xx_0$\,.
  \end{itemize}
\end{definition}

\subsubsection{Polynomial Vector Field and Lie Derivatives}
In this paper, we focus on vector fields defined by polynomials.
\begin{definition}[Polynomial Vector Field] \label{def:plvf}
Suppose $\fb=(f_1,f_2,\cdots,f_n)$  in (\ref{eq:ode}). If for all $1\leq i\leq n$,  $f_i$ is a polynomial in
$\mathbb Q[x_1,x_2,\ldots,x_n]$,  then $\fb$ is called a polynomial vector field, denoted by
$\fb\in \mathbb Q^n[\xx]$.
\end{definition}

Obviously polynomial vector fields satisfy the local  Lipschitz condition. Let $p$ be a polynomial in ring $\mathbb Q[\xx]$, which is a scalar function. Then the gradient of $p$ : $$ \frac{\partial}{\partial \xx} p \,\,\define\,\, (\frac{\partial p}{\partial x_1},\frac{\partial p}{\partial x_2},\cdots, \frac{\partial p}{\partial x_n})$$
is a vector of polynomials with dimension $\dim{(\xx)}$ . Thus the inner product of a polynomial vector field $\fb$ and the gradient of a polynomial $p$ is still a polynomial, if $\fb\in \mathbb Q^n[\xx]$ and
$\dim{(\xx)}=n$ (in the rest of the paper, this will be assumed implicitly). Therefore we can inductively define the {\itshape Lie derivatives} of $p$ along $\fb$, $L^k_{\fb} p: \mathbb R^n
\mapsto \mathbb R$, for $k \in \mathbb N$, as follows:
\begin{itemize}
\item $L^0_{\fb} p(\xx)=p(\xx)$,
\item $L^k_{\fb} p(\xx)=(\frac{\partial}{\partial \xx} L^{k- 1}_{\fb} p(\xx),
\fb(\xx)$), \,for $k>0$,
\end{itemize}
where $(\cdot, \cdot)$ is the inner product of two vectors, that is, $(\mathbf{a},\mathbf{b})=\sum_{i=1}^{n}a_ib_i$ for
$\mathbf{a}=(a_1,\ldots,a_n)$ and $\mathbf{b}=( b_1,\ldots, b_n )$.
\begin{example}\label{eg:Lie-derv}
Suppose $\fb=(-x,y)$ and $p(x,y)=x+y^2$,  then
\begin{displaymath}
\begin{array}{l}
L^0_{\fb} p=x+y^2\\
L^1_{\fb} p=-x+2y^2\\L^2_{\fb} p=x+4y^2\\
\cdots \quad\,\,\cdots
\end{array}
\end{displaymath}
\end{example}
For a parametric  polynomial $p(\uu,\xx)$, we can define the Lie derivatives of $p$ along $\fb$ similarly if the gradient of $p(\uu,\xx)$ is taken as $ \frac{\partial}{\partial \xx} p (\uu,\xx)$, and all $L^i_{\fb} p(\uu,\xx)$ are still parametric polynomials.

Given a polynomial vector field, we can make use of Lie derivatives to investigate the tendency of its trajectory in terms of a polynomial $p$ (as an energy function). To capture this, look at Example \ref{eg:Lie-derv} shown in I of Figure \ref{fig:intu-lie}.

In I of Figure \ref{fig:intu-lie}, arrow $B$ denote the corresponding evolution direction according to the vector field  $\fb=(-x,y)$, and we could imagine the points on the parabola $p(x,y)=x+y^2$ with  zero energy, and the points in white area have positive energy, i.e., $p(x,y)>0$. Arrow $A$ is the gradient  $\frac{\partial}{\partial \xx} p|_{(-1,1)}$ of $p(x,y)$, which infers that the trajectory starting at $(-1,1)$ will enter white area immediately if the angle, between $\frac{\partial}{\partial \xx} p|_{(-1,1)}$ and the evolution direction at $(-1,1)$, is less than $\frac{\pi}{2}$, that is, the 1-order Lie derivative is positive. Thus the 1-order Lie derivative $L^1_{\fb} p|_{(-1,1)}=3$ of $p$ along $\fb$ (the inner product of  $\frac{\partial}{\partial \xx} p|_{(-1,1)}$ and $\fb(x,y)|_{(-1,1)}$) predicts that there is some positive $d>0$ such that the trajectory starting at $(-1,1)$ (curve $C$) has the property $p(\xx((-1,1),t))>0$ for all $t\in (0,d)$.

However, if the angle between gradient and evolution direction is $\frac{\pi}{2}$ or the gradient is zero-vector, then 1-order Lie derivative is zero and it is impossible to predict trajectory tendency by means of 1-order Lie derivative.  In this case, we resort to nonzero higher order Lie derivatives. For this purpose, we introduce
the {\itshape pointwise rank} of $p$ with respect to $\fb$ as the function $\gamma_{p,\fb}: \mathbb R^n \mapsto  \mathbb N\cup \{\infty\}$ defined by
\[\gamma_{p,\fb}(\xx)=\min\{k\in \mathbb N\mid L^k_{\fb}p(\xx)\neq 0\},\]
if such $k$ exists, otherwise $\gamma_{p,\fb}(\xx)=\infty$.

\begin{example}\label{eg:rank}
Let \, $\fb(x, y)=(\dot{x}=-2y, \dot{y}=x^2)$ and
$h(x, y)=x+y^2$,  then
\begin{eqnarray*}
L^0_{\fb} h(x, y)& = & x+y^2 \\
L^1_{\fb} h(x, y)& = & -2y+2x^2y \\
L^2_{\fb} h(x, y)& = & -8y^2x-(2-2x^2)x^2 \\
 & \vdots &
\end{eqnarray*}
Here, $\gamma_{h,\fb}(0, 0)=\infty$, $\gamma_{h,\fb}(-4, 2)=1$, etc.

Look at  II of Figure \ref{fig:intu-lie}.  At point $(-1, 1)$ on curve $h(x,y)=0$, the gradient of $h$ is $(1, 2)$ (arrow $A$) and the evolution direction is $(-2,1)$ (arrow B), so their inner product is zero. Thus it is impossible to predict the tendency (in terms of curve $h(x,y)=0$) of trajectory starting from $(-1, 1)$ via its 1-order Lie derivative. By a simple computation, its 2-order Lie derivative is $8$. Hence $\gamma_{h,\fb}(-1, 1)=2$. In the sequel, we shall show how to use such high order Lie derivatives to analyze the trajectory tendency.

\end{example}

\begin{figure}[tp]
\begin{center}
\includegraphics[width=1.5in,height=1.5in]{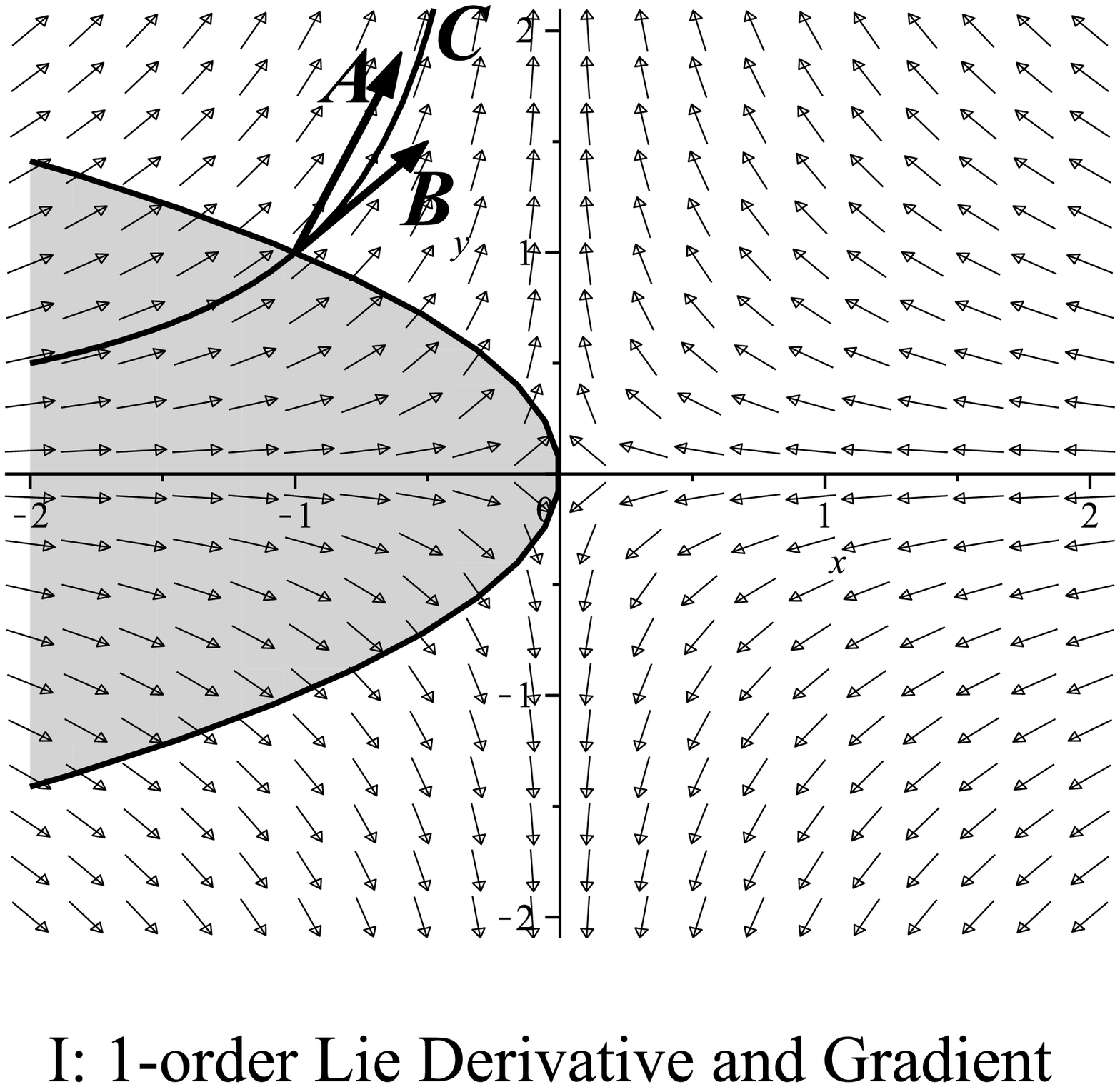}
\includegraphics[width=1.5in,height=1.5in]{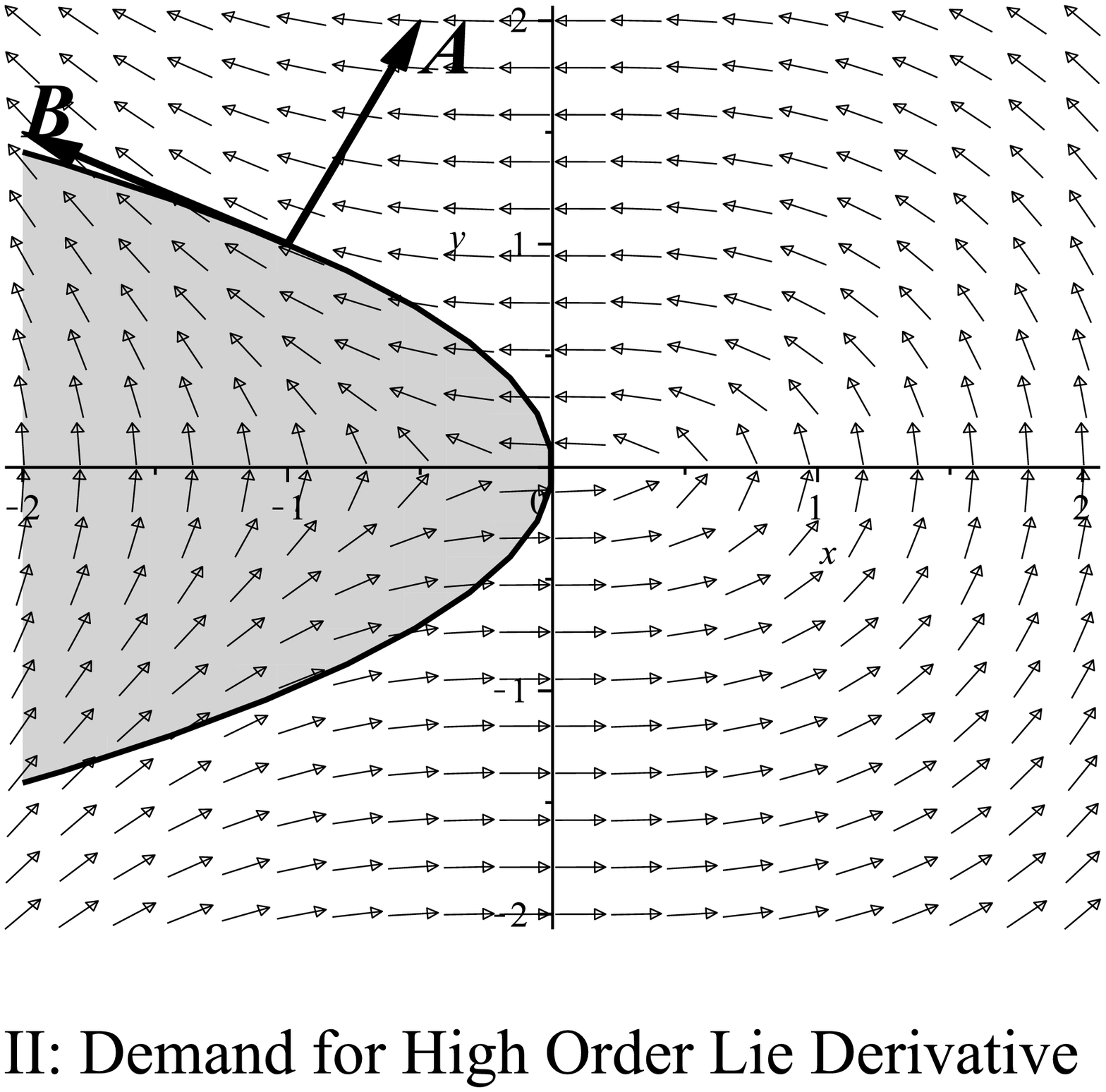}
\end{center}
\caption{Lie Derivatives}
\label{fig:intu-lie}
\end{figure}

For analyzing trajectory tendency by high order Lie derivatives, we need the following fact.
\begin{proposition}\label{Trans} Given polynomial functions $p$ and $\fb$, then $\xx_0$ is on the boundary $\mathcal S(p(\xx)=0)$ iff $\gamma_{p,\fb}(\xx_0) \neq 0$. Suppose $\xx_0=\xx(0)$, then it follows that
\begin{enumerate}
\item[(a)] if $\gamma_{p,\fb}(\xx_0)< \infty$ and $L^{\gamma_{p,\fb}(\xx_0)}_{\fb} p(\xx_0)>0$, then
 \[\exists \epsilon >0,\forall t\in
(0,\epsilon). p(\xx(t))> 0;\]

\item[(b)] if $\gamma_{p,\fb}(\xx_0)< \infty$ and $L^{\gamma_{p,\fb}(\xx_0)}_{\fb} p(\xx_0)<0$, then
\[\exists \epsilon >0, \forall t\in (0, \epsilon). p(\xx(t))< 0;\]

\item[(c)] if $\gamma_{p,\fb}(\xx_0)= \infty$, then
\[\exists \epsilon >0, \forall t\in (0, \epsilon). p(\xx(t))=0.\]

\end{enumerate}
\end{proposition}

\begin{proof}
Polynomial functions are analytic, so $\fb$ is analytic and thus
$\xx(t)$ is analytic in a small interval $(a,b)$ containing zero \cite{Pollard}. Besides, $p$ is analytic, so the Taylor expansion of $p(\xx(t))$ at $t=0$
\begin{eqnarray}
p(\xx(t))& = &p(\xx_0)+\frac{\ud p}{\ud t}\cdot t + \frac{\ud^{2} p}{\ud t^{2}}\cdot\frac{t^{2}}{2!}+\cdots \nonumber\\
& = & L_{\fb}^0 p(\xx_0)+L_{\fb}^1 p(\xx_0)\cdot t + L_{\fb}^2 p(\xx_0)\cdot \frac{t^{2}}{2!}+\cdots\quad\label{eqn:taylor}
\end{eqnarray}
converges in another small interval $(a',b')$ containing zero \cite{Krantz}. Then the conclusion of Proposition \ref{Trans} follows immediately from formula (\ref{eqn:taylor})  by case analysis on the sign of $L_{\fb}^{\gamma_{p,\fb}(\xx_0)}p(\xx_0)$.
\end{proof}

Based on this proposition, we introduce the notion of {\itshape transverse set} to indicate the tendency of the trajectories of a considered polynomial vector field in terms of the first nonzero Lie derivative of a underlying polynomial as follows.
\begin{definition}\label{Trans-def}
Given a polynomial $p$ and a polynomial vector field $\dot{\xx}=\fb(\xx)$, the transverse set of $\fb$ over the domain $\mathcal S(p(\xx)\geq 0)$ is
\[Trans_{\fb\uparrow p} \define  \{\xx\in \mathbb R^n\mid \gamma_{p,\fb}(\xx)<\infty \, \wedge \,L^{\gamma_{p,\fb}(\xx)}_{\fb}
p(\xx)<0 \}.\]
\end{definition}

Intuitively, if $\xx\in Trans_{\fb\uparrow p}$, then either $\xx$ is  not in  $\mathcal S(p(\xx)\geq 0)$ or $\xx$ is on the boundary of  $\mathcal S(p(\xx)\geq 0)$ such that the trajectory $\xx(t)$ starting with $\xx$ will exit $\mathcal S(p(\xx)\geq 0)$ immediately.

\section{Semi-algebraic Invariant}\label{sec:di}
A hybrid system consists of a set of CDSs, a set of jumps between these CDSs, and a set of initial states. The CDSs in a hybrid system are a little different from the standard ones, as normally they are equipped with a domain and a set of initial states, in the form $(H,\fb,\Xi)$, where $H$ is used to force some jumps outgoing the mode to happen, that is, a hybrid system can stay within a mode only if the domain of the current mode holds, and $\Xi$ is a subset of $H$, standing for the set of initial states. Obviously, a CDS can be seen as a  special  CDSwDI by letting $H=\RR^n$.
The goal of this paper is to present a complete method for automatically discovering SAIs of PDSs, based on which, as we discussed in the introduction, we can finally verify polynomial hybrid systems.

\subsection{Continuous Invariants of CDSwDI}
The notion of  \emph{continuous invariant}
of {CDSwDI} is quite similar to the one of positive invariant set of CDS \cite{Blanchini}.
Informally, a continuous invariant $P$ of a CDwDI $(H,\fb,\Xi)$ is a superset of $\Xi$ such that all continuous evolutions starting from $\Xi$ keep within $P$ if they are within $H$. Here, we give a formal definition of CI adapted from \cite{PlatzerClarke08} as follows:
\begin{definition}[Continuous Invariant \cite{PlatzerClarke08}]\label{dfn:inv} Given a CDSwDI $(H,\fb,\Xi)$ with $\Xi\subseteq H\subseteq \RR^n$ and  $\fb:\RR^n\mapsto\RR^n$ that is local Lipschitz continuous, a set $P\subseteq\RR^n$ is called a \emph{continuous invariant} of $(H,\fb, \Xi)$, iff
\begin{enumerate}
\item $\Xi \rightarrow P$; and
 \item for all $\xx_0\in P$, and for any $T\geq 0$,
 $$ (\forall t \in [0,T]. \xx(\xx_0;t)\in H) \rightarrow (\forall t \in [0,T]. \xx(\xx_0;t)\in P).$$
 \end{enumerate}
\end{definition}

Regarding Definition \ref{dfn:inv}, we would like to give the following remarks.

\begin{enumerate}
  \item Continuous invariant in Definition \ref{dfn:inv} is more general than standard positive invariant set of continuous dynamical systems. However, if $H=\RR^n$ and $\Xi=P$, then the two notions coincide.

\item One may have noticed that in Definition \ref{dfn:inv}, a continuous invariant set $P$ is not necessarily a subset of domain $H$. In fact, any $P$ satisfying  $H\subseteq P$ is continuous invariant of $(H,\fb,\Xi)$. This seems weird at first sight, because such continuous invariant sets are useless if we only concern the CDSwDI in isolation.
     But it would be quite useful in the verification of the hybrid system if we  assume that the continuous invariant of a mode always holds if the hybrid system does not stay within the mode.
\end{enumerate}

\subsection{PDS and SAI}

\begin{definition}\label{dfn:pds}
A CDSwDI $(H,\fb,\Xi)$ is called a polynomial dynamical system with semi-algebraic domain and initial states (PDS), if $H$ and $\Xi$ are semi-algebraic sets  and $\fb$ is a polynomial vector field in  $\QQ^n[\xx]$.

A continuous invariant of a PDS is called a \emph{semi-algebraic invariant} (SAI) if it is a semi-algebraic set.
\end{definition}

In the subsequent sections, we will present a sound and complete method to automatically discover SAIs for a PDS.

\section{Fundamental Results}\label{sec:fund}
The set $Trans_{\fb\uparrow p}$ in Definition \ref{Trans-def} plays a crucial role in our theory. First of all, we have
\begin{theorem}\label{Comp-Trans}
The set $Trans_{\fb\uparrow p}$ is a semi-algebraic set if $p$ is a  polynomial and $\fb$ is a polynomial vector field, and hence it is computable.
\end{theorem}

To prove this theorem, it suffices to show $\gamma_{p,\fb}(\xx)$ is computable for each $\xx\in  \mathcal S(p(\xx)\geq 0)$. However, $\gamma_{p,\fb}(\xx)$ may be infinite for some $\xx\in \mathcal S(p(\xx)\geq 0)$. Thus, it seems that we have to compute $L^{k}_{\fb} p(\xx)$ infinite times for such $\xx$ to determine if $\xx\in Trans_{\fb\uparrow p}$. Fortunately, we can find a uniform upper bound on $\gamma_{p,\fb}(\xx)$ for all $\xx$ with $\gamma_{p,\fb}(\xx)$ being finite.

\begin{theorem}[Rank Theorem]\label{Rank}
If $p$ and $\fb$ are polynomial functions, then there is an integer $N$ such that for all $\xx\in \mathbb R^n$, $\gamma_{p,\fb}(\xx)<\infty$ implies  $\gamma_{p,\fb}(\xx)\leq N$. Later on, such an $N$ is called the {\itshape rank} of $p$ and $\fb$, denoted by $\gamma_{p, \fb}$.
\end{theorem}

\begin{proof}
Let $D_l=\{\xx\mid \forall m<l. L^m_{\fb} p(\xx)=0\}$ for $l\geq 0$. Note that the sequence $\{D_l\}_{l\in \mathbb N}$ is decreasing. We will show that there is an $N$ such that $D_l=D_N$ for all $l\geq N$.

Since $p$ and $\fb$ are polynomial functions, all $L^m_{\fb}p(\xx)$ must be polynomials for any $m\in \mathbb N$. We consider the polynomial ideal $I$ generated by $\{L^m_{\fb}p(\xx)\mid m\in \mathbb N\}$. Let $I_m=\< L^0_{\fb}p(\xx), L^1_{\fb}p(\xx), \cdots, L^m_{\fb}p(\xx)\rg$. Then $I=\cup_m I_n$. By Theorem \ref{ACC}, there is $k$ such that $I=I_k$. Thus for all $l>k$, there are $g_i\in \mathbb R[x_1,  \cdots, x_n]$ for $i\leq k$ such that $L^l_{\fb}p(\xx)=\sum_{i\leq k}g_i L^i_{\fb}p(\xx)$ for all $\xx\in \mathbb R^n$.

Fix $l>k$. If $\xx\in D_l$, then $L^l_{\fb}p(\xx)=\sum_{i\leq k}g_i L^i_{\fb}p(\xx)=0$ since all $L^i_{\fb}p(\xx)=0$ for $i\leq k$ as $\xx\in D_l$. Let $N=k+1$. Then  $D_l=D_N$ for all $l\geq N$. Thus, if $\xx\in D_N$ then $\gamma_{p,\fb}(\xx)=\infty$. Therefore, $\gamma_{p,\fb}(\xx)<\infty$ implies $\gamma_{p,\fb}(\xx)\leq N$.
\end{proof}

Now, it suffices to compute the values
\[L^{0}_{\fb}p(\xx_0), L^{1}_{\fb}p(\xx_0) \cdots, L^{\gamma_{p, \fb}}_{\fb}p(\xx_0)\]
to determine whether $\gamma_{p, \fb}(\xx_0)$ is infinite. Therefore if $\gamma_{p, \fb}$ is computable then $Trans_{\fb\uparrow p}$ is computable too. It is desirable to get an expression of $\gamma_{p, \fb}$ for given $p$ and $\fb$. However, we did not find it yet. Nevertheless, a computable upper bound for $\gamma_{p, \fb}$ can indeed be found effectively according to the following theorem.
\begin{theorem}[Fixed Point Theorem]\label{FP}
If \[L^{i+1}_{\fb}p\in \<L^{0}_{\fb}p, L^{1}_{\fb}p,  \cdots, L^{i}_{\fb}p\rg,\]
then $L^{m}_{\fb}p\in \<L^{0}_{\fb}p, L^{1}_{\fb}p, \cdots, L^{i}_{\fb}p\rg$, for all $m>i$.
\end{theorem}

\begin{proof}
We prove this theorem by induction. Assume this conclusion is true for all $l\leq k$ with $k> i$. Especially, $L^{k}_{\fb}p\in \<L^{0}_{\fb}p, L^{1}_{\fb}p, \cdots, L^{i}_{\fb}p\rg$. Then there are $g_j\in \mathbb R[x_1, \cdots, x_n]$ for $j\leq i$ such that \begin{equation}\label{Induction}
L^{k}_{\fb}p=\sum_{j\leq  i} g_jL^{j}_{\fb}p.
\end{equation}
By the definition of Lie derivative and equation (\ref{Induction}), it follows that
\begin{align}
& L^{k+1}_{\fb} p \nonumber \\
~&~=~(\frac{\partial}{\partial \xx}L^k_{\fb}p, \fb) \nonumber\\
~&~=~(\frac{\partial}{\partial \xx}(\sum_{j\leq i} g_jL^{j}_{\fb}p), \fb) \nonumber\\
~&~=~\sum_{j\leq i}(L^{j}_{\fb} p \frac{\partial}{\partial
\xx}g_j,\fb)+
 \sum_{j\leq i}(g_j \frac{\partial}{\partial \xx}L^{j}_{\fb}p, \fb)
 \nonumber \\
~&~=~\sum_{j\leq i}(\frac{\partial}{\partial \xx}g_j, \fb)L^{j}_{\fb} p+\sum_{j\leq i}g_jL^{j+1}_{\fb}p\nonumber \\
~&~=~\sum_{j\leq i}(\frac{\partial}{\partial \xx}g_j, \fb)L^{j}_{\fb}
p+\sum_{j< i}g_jL^{j+1}_{\fb}p+g_{i}L^i_{\fb}p. \nonumber
\end{align}
By induction hypothesis, $L^i_{\fb}p$ is in $\<L^{0}_{\fb}p,
L^{1}_{\fb}p, \cdots, L^{i}_{\fb}p\rg$. So
$$L^{k+1}_{\fb}p\in\<L^{0}_{\fb}p, L^{1}_{\fb}p, \cdots,
L^{i}_{\fb}p\rg.$$ By induction, the theorem follows immediately.
\end{proof}

Let $N_{p, \fb}$ be the minimal $i$ satisfying the condition of Theorem~\ref{FP} in the sequel. Then  $\gamma_{p, \fb}\leq N_{p, \fb}$.  Look at Example \ref{eg:rank}, where $N_{h, \fb}=2$. Now, applying above two theorems we can prove Theorem \ref{Comp-Trans}.

\begin{proof}[of Theorem \ref{Comp-Trans}]
Since $\gamma_{p, \fb}\leq N_{p, \fb}$,
\[\xx\in Trans_{\fb\uparrow p} \mbox{ iff }
  \gamma_{p, \fb}(\xx)\leq N_{p, \fb} \, \wedge  \, L^{\gamma_{p,\fb}(\xx)}_{\fb}
p(\xx)<0. \]
Therefore, $Trans_{\fb\uparrow p}$ is computable as $N_{p, \fb}$ is computable according to Theorem \ref{FP}. Given $p$ and $\fb$, let
\begin{equation}
\pi^{(0)}(p,\fb,\xx)\, \define \,  ~ p(\xx)<0,
\end{equation}
for $1\leq i\in \mathbb N$,
\begin{equation}
\pi^{(i)}(p,\fb,\xx) ~ \define ~ \left(\bigwedge_{0\leq j<
i}L^{j}_{\fb} p(\xx)=0\right) \wedge L^{i}_{\fb} p(\xx)<0,
\end{equation}
and
\begin{eqnarray}
\pi(p, \fb, \xx)& \define &
                        \bigvee_{0\leq i\leq
N_{p,\fb}}\pi^{(i)}(p,\fb,\xx).
\end{eqnarray}
By Theorem \ref{Rank} and $\gamma_{p, \fb}\leq N_{p, \fb}$, we have another equivalence
\begin{equation}\label{Trans-eqv}
\xx\in Trans_{\fb\uparrow p} \mbox{ iff } \pi(p, \fb, \xx) \text{holds}.
\end{equation}

In fact, $\pi^{(i)}(p,\fb,\xx)$ here is a particular semi-algebraic
system, and so $\pi(p, \fb, \xx)$ is a union of semi-algebraic systems. Thus $Trans_{\fb\uparrow p}$ is actually a semi-algebraic set.
\end{proof}

In the SAI generation, it actually makes use of parametric polynomials $p(\uu,\xx)$ with parameter $\uu=(u_1,u_2,\dots,$ $ u_t)$.
The following theorem indicates Theorem~\ref{Rank} still holds after substituting $p(\uu,\xx)$ for $p(\xx)$.

\begin{theorem}[Parametric Rank Theorem]\label{Para-Rank}
Given \, polynomial functions $p(\uu,\xx)$ and $\fb$, there is an integer $N$ such that $\gamma_{p_{\uu_0},\fb}(\xx)<\infty$ implies $\gamma_{p_{\uu_0},\fb}(\xx)\leq N$ for all $\xx\in \mathbb R^n$ and all $\uu_0\in \mathbb{R}^t$.
\end{theorem}
This proof is quite close to the one of Theorem \ref{Rank}. The difference, between the proof of this theorem and the one of Theorem \ref{Rank}, lies in the settings of polynomials. Here, we consider polynomials $p$ and $\fb$ in the polynomial ring $\mathbb R[\uu, \xx]$. Similarly, we also introduce the rank function on polynomials with parameters, still denoted by $\gamma_{p, \fb}$. Accordingly, let  $N_{p,\fb}$ denote the upper bound  computed by a similarity of Theorem~\ref{FP} .

\section{Generating SAI in  Simple Case}\label{sec:idea}

Given a polynomial vector field $\dot{\xx}=\fb(\xx)$ with a semi-algebraic domain $H$ and initial condition $\Xi$, our task is to find a semi-algebraic set $P$ such that $P$ is an SAI of $(H,\fb, \Xi)$.

First of all, we illustrate our idea by showing how to compute an SAI of the simple form $P\, \define \,  p(\xx)\geq 0$ for a simple domain $H \define \, h(\xx)\geq 0$. For convenience, we will simply write the dynamical system $(h(\xx)\geq 0, \fb, \Xi)$ as $(h,\fb, \Xi)$. Notice that $P$ is an SAI of  $(h,\fb, \Xi)$ only if $\forall \xx(\Xi(\xx)\rightarrow P(\xx))$. It is evident that if $\xx(0)$ is in the interior of $\mathcal S(p(\xx)\geq 0) \cap \mathcal S(h(\xx)\geq 0)$, then the trajectory $\xx(t)$ starting at $\xx(0)$ will remain in the interior within adequately small $t>0$. Therefore, the condition of continuous invariant could be violated only at the points on the boundary $\mathcal S(p(\xx)=0) $ of $\mathcal S(p(\xx)\geq 0) $. Thus by Definition \ref{Trans-def} and Proposition \ref{Trans}, $p\geq 0$ is an invariant of $(h, \fb, \Xi)$ if and only if it meets $\forall \xx(\Xi(\xx)\rightarrow P(\xx))$ and
\begin{equation*}
    \xx\in \mathcal S(p(\xx)=0) \rightarrow \xx\notin Trans_{\fb\uparrow p}
    \setminus Trans_{\fb\uparrow h},
\end{equation*}
i.e.
\begin{equation}\label{Inv}
    \xx\in \mathcal S(p(\xx)=0) \rightarrow \xx \in (Trans_{\fb\uparrow p})^c
    \vee Trans_{\fb\uparrow h}.
\end{equation}

By equivalences (\ref{Trans-eqv}), the formula (\ref{Inv}) is equivalent to
\begin{equation*}
 p(\xx)= 0  \rightarrow (\neg\pi(p, \fb, \xx) \vee \pi(h, \fb, \xx)),
\end{equation*}
i.e.
\begin{equation}\label{equiv}
 \big(p(\xx)= 0 \wedge \pi(p, \fb, \xx)\big) \rightarrow  \pi(h, \fb, \xx).
\end{equation}

Let $\theta(h,p,\fb, \xx)$ denote the formula (\ref{equiv}). According to this equivalence, we obtain the sufficient and necessary condition for being SAI as follows.
\begin{theorem}[Criterion Theorem]\label{thm:crit}
Given a polynomial $p$, $p(\xx)\geq 0$ is an SAI of system $(h, \fb, \Xi)$ if and only if the formula $\theta(h,p,\fb, \xx) \wedge (\Xi(\xx)\rightarrow p(\xx)\geq 0)$ is true for all $\xx\in \mathbb R^n$.
\end{theorem}

Now, we are ready to present a constraint based approach to generate polynomial continuous invariants. The basic idea is as follows:
\begin{enumerate}
\item[I.] First, set a parametric polynomial $p$ as
\begin{equation}\label{Template}
p(\uu,\xx)\, \define \, \sum_{i_1+i_2+\cdots+i_n=k\leq d}u_{i_1i_2\cdots i_n}x_1^{i_1}x_2^{i_2}\cdots x_n^{i_n}.
\end{equation}
Such a parametric polynomial is called a {\it template} conventionally. There are $t= {n+d \choose d}$ many terms and accordingly $t$ many parameters $u_{i_1i_2\cdots i_n}$. For simplicity, let $\uu$ denote such a $t$-tuple $\{u_{i_1i_2\cdots i_n}\}_{i_1+i_2+\cdots+i_n=k\leq d}$.

\item[II.] Then we appy the quantifier elimination (QE\footnotemark \footnotetext{QE has been implemented in many computer algebra tools such as DISCOVERER \cite{Xia07}, QEPCAD \cite{qepcad} and Redlog \cite{Redlog}.}  for short) to the formula $\forall \xx. (\theta(h,p,\fb,\xx) \wedge (\Xi(\xx)\rightarrow p(\xx)\geq 0))$. If the output is {\it false}, then there is no polynomial continuous invariant of degree $\leq d$ for $(h, \fb,\Xi)$. Otherwise, it will give us a constraint on $\uu$, denoted by $R(\uu)$. In fact, $R(\uu)$ is a union of
semi-algebraic systems (refer to \cite{tarski51}).

\item[III.] Let $S_{\textit{Inv}}$ be the set of solutions to $R(\uu)$. Now, using a tool like DISCOVERER \cite{Xia07} to pick a $\uu_0\in S_{\textit{Inv}}$ and then $p_{\uu_0}(\xx)\geq 0$ is an invariant of $(h, \fb,\Xi)$ by Theorem \ref{thm:crit}.
\end{enumerate}
\begin{remark}
\begin{itemize}
\item[1)] Note that in real applications, one usually picks up the specific terms with nonzero coefficients. A simplified template could make the resulted polynomial satisfy special conditions and also reduce the complexity of the searching process.
\item[2)] In the above Step III, if the dimension of $S_{\textit{Inv}}$ equals $t$, then we can easily select a rational sample point $\uu_0$ from $S_{\textit{Inv}}$ and the obtained $p_{\uu_0}(\xx)\geq 0$ is an SAI in $\mathbb R^n$; otherwise when it is difficult (or impossible) to get a rational instantiation for $\uu$, we can always compute an algebraic sample point $\uu_0\in S_{\textit{Inv}}$, that is, $\uu_0$ is itself defined by polynomial equations. It is easy to show that in the latter case, $p_{\uu_0}(\xx)\geq 0$ is also an SAI in $\mathbb R^{n}$.
\end{itemize}
\end{remark}
\begin{example} Again, we make use of Example \ref{eg:rank} to demonstrate above method. That is, $\fb(x, y)\define(\dot{x}=-2y, \dot{y}=x^2)$. Here, we take $H\define\{(x,y)\in \mathbb R^2\mid h(x, y)=-x-y^2\geq 0\}$ as the domain and $\Xi\define \{(-1, 0.5), (-0.5, -0.6)\}$ as the initial states. Apply procedure (I-III), we have:
\begin{enumerate}
\item Set a template $p(\uu,\xx):=ay(x-y)\geq 0$ where $\uu\define \langle a\rangle$. Then we have $\gamma_{p,\fb}\leq N_{p,\fb}=2$.
\item Compute the corresponding formula
\begin{eqnarray*}
\theta(h,p,\fb,\xx)& \define & p=0 \wedge (\pi^{(0)}_{p,\fb,\xx}\vee \pi^{(1)}_{p,\fb,\xx}\vee \pi^{(2)}_{p,\fb,\xx}) \rightarrow \\& & (\pi^{(0)}_{h,\fb,\xx} \vee \pi^{(1)}_{h,\fb,\xx}\vee \pi^{(2)}_{h,\fb,\xx} )
\end{eqnarray*}
where
\begin{eqnarray*}
\pi^{(0)}_{h,\fb,\xx}& \define & -x-y^2<0, \\
\pi^{(1)}_{h,\fb,\xx}& \define & -x-y^2=0 \wedge 2y-2x^2y<0, \\
\pi^{(2)}_{h,\fb,\xx}& \define & -x-y^2=0 \wedge 2y-2x^2y=0 \, \wedge \\
& &  8xy^2+2x^2-2x^4<0 , \\
\pi^{(0)}_{p,\fb,\xx} & \define & ay(x-y)<0 , \\
\pi^{(1)}_{p,\fb,\xx} & \define & ay(x-y)=0\wedge -2ay^2+ax^3-2yax^2 <0 , \\
\pi^{(2)}_{p,\fb,\xx} & \define &  ay(x-y)=0\wedge -2ay^2+ax^3-2yax^2 =0\, \wedge \\ & & 40axy^2-16ay^3+32ax^3y-10ax^4<0.
\end{eqnarray*}
Then we implement quantifier elimination on formula $\forall x,y (\theta(h,p,\fb,\xx)\wedge (0.5a(-1-0.5)\geq 0 \wedge -0.6a(-0.5+0.6)\geq0)$. We get the constraint on $a$ is $a\leq 0$
\item Just pick $a=-1$, and then $-xy+y^2\geq 0$ is an invariant for $(H, \fb,\Xi)$. The grey part of the picture III is the intersection of this invariant and domain $H$.
\begin{figure}[tp]
\begin{center}
\includegraphics[width=1.5in,height=1.5in]{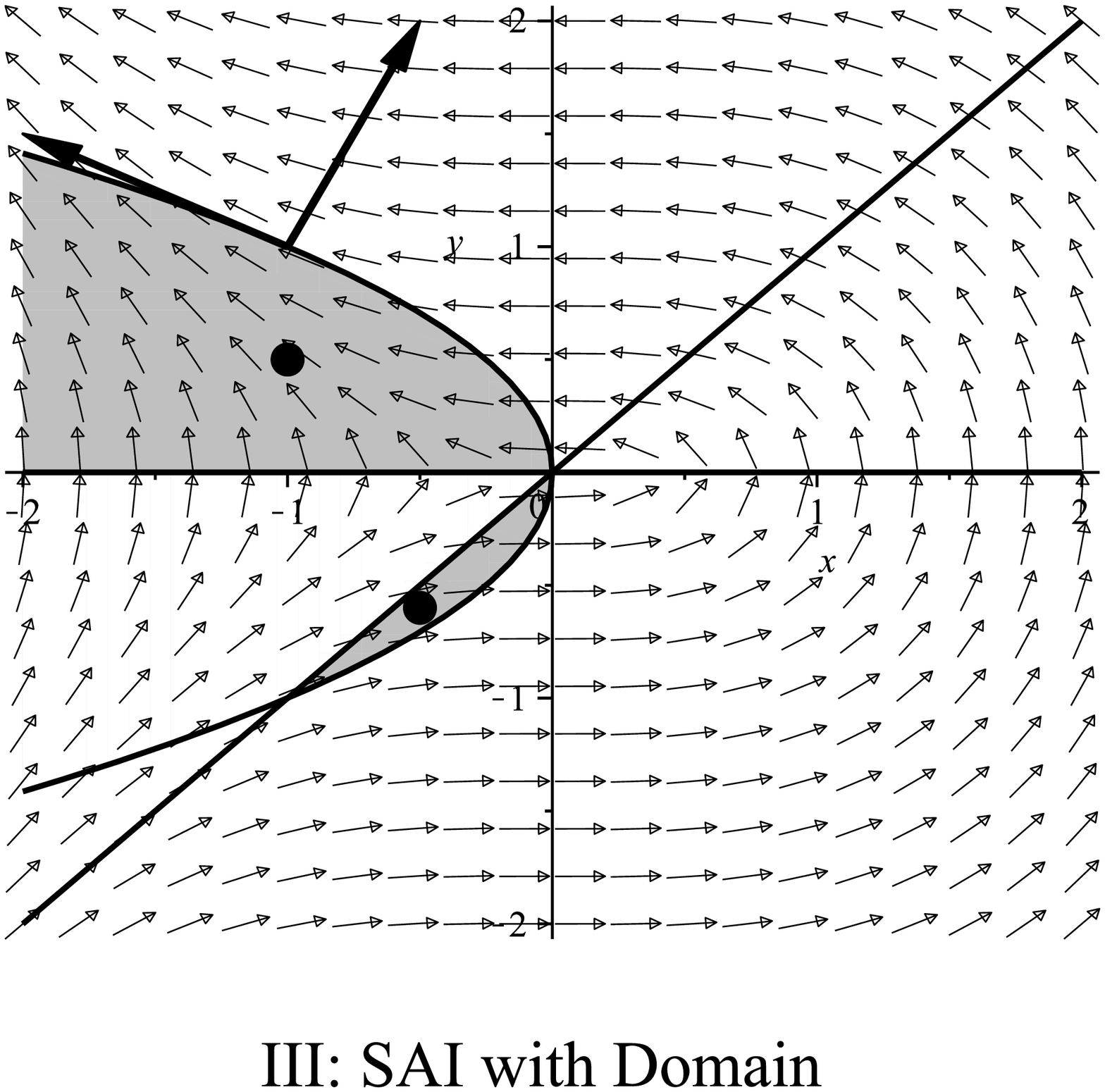}
\includegraphics[width=1.5in,height=1.5in]{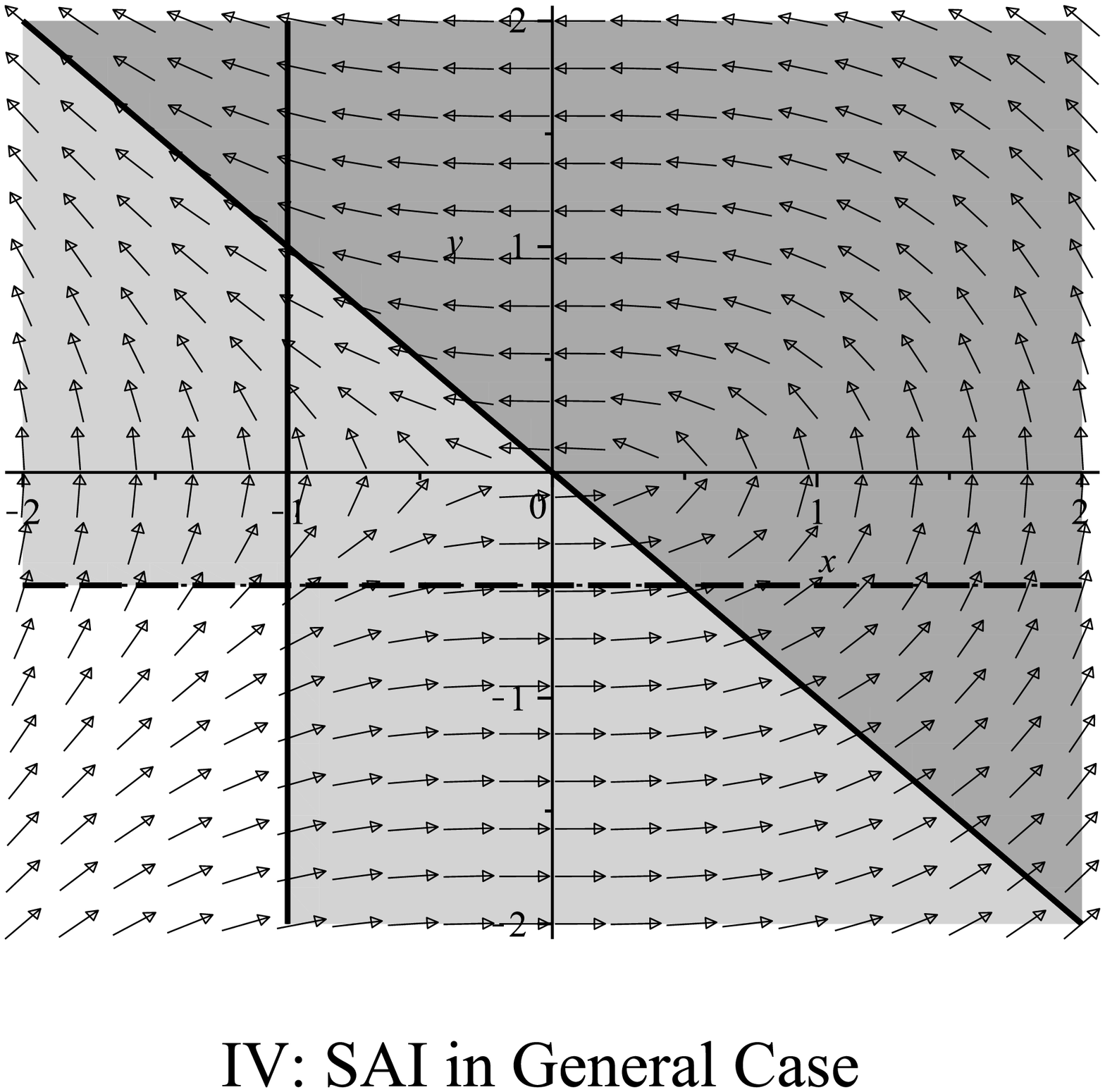}
\caption{Semi-Algebraic Invariants}
\label{fig:Inv}
\end{center}
\end{figure}

\end{enumerate}

\end{example}

\section{General Case}\label{sec:AIG}
Now, we consider how to automatically discover SAIs of a PDS in general case. Given a PDS $(H,\fb, \Xi)$ with
\begin{equation}\label{for:dom}
H=\mathcal S(\bigvee_{i=1}^{I}\bigwedge_{j=1}^{J_i}p_{ij}(\xx)\triangleright 0), \, \,
 \Xi = \mathcal S(\bigvee_{i=1}^{N}\bigwedge_{j=1}^{M_i}q_{ij}(\xx)\triangleright 0)
\end{equation}
and $\fb\in\QQ^n[\xx]$, where {$\Xi \subseteq H$} and $\triangleright\in \{\geq,>\}$. The procedure for automatically generating SAIs with a general template
 \begin{equation*}
   P=\mathcal S(\bigvee_{k=1}^{K}\bigwedge_{l=1}^{L_k}p_{kl}(\uu_{kl},\xx)\triangleright 0)\, ,\,\mbox{where} \, \triangleright\in \{\geq,>\}
  \end{equation*}
for $(H,\fb, \Xi)$, is essentially the same as the steps (I-III) depicted in the previous section.  However, we must sophisticatedly handle the complex combination due to the complicated boundaries. In what follows, we will first establish the necessary and sufficient conditions for general CIs of  a CDSwDI by some topological analysis. Then we show for SAIs of a PDS, these conditions can be encoded equivalently into first order polynomial formulas.

\subsection{Necessary and Sufficient Condition for CI}
First of all, we study a necessary and sufficient condition like formula (\ref{Inv}) for $P$ being an invariant of $(H,\fb, \Xi)$. To analyze the evolution tendency of trajectories dominated by a locally Lipschitz continuous vector field $\fb:\RR^n\mapsto\RR^n$ in terms of a subset $A$ of $\RR^n$, we need the following notions and notations.
\begin{align*}
\mbox{In}_{\fb}(A) &\,\define\,\{\xx_0\in\RR^n\mid\exists\epsilon>0 \forall t \in(0,\epsilon). \,\xx(\xx_0;t)\in A\},\\
\mbox{IvIn}_{\fb}(A) &\,\define\,\{\xx_0\in\RR^n\mid\exists\epsilon>0 \forall t \in(0,\epsilon). \,\xx(\xx_0;-t)\in A\}.
\end{align*}

Intuitively, $\xx_0\in \mbox{In}_{\fb}(A)$ means that the trajectory starting from $\xx_0$ enters $A$ immediately and keeps inside $A$ for some time; $\xx_0\in \mbox{IvIn}_{\fb}(A)$ means that the trajectory through $\xx_0$ reaches $\xx_0$ from the interior of $A$.

Analogous to $\mbox{In}_{\fb}(A)$ and $\mbox{IvIn}_{\fb}(A)$, we introduce another two notations $\mbox{Out}_{\fb}(A)$ and $\mbox{IvOut}_{\fb}(A)$.
\begin{align*}
\mbox{Out}_{\fb}(A) &\,\define\, \{\xx_0\in\RR^n\mid\exists\epsilon>0 \forall t \in(0,\epsilon). \,\xx(\xx_0;t)\in A^c\};\\
\mbox{IvOut}_{\fb}(A) &\,\define\, \{\xx_0\in\RR^n\mid\exists\epsilon>0 \forall t \in(0,\epsilon). \,\xx(\xx_0;-t)\in A^c\},
\end{align*}
where $A^c$ stands for the complement of $A$ in $\RR^n$. Intuitively, $\xx_0\in \mbox{Out}_{\fb}(A)$ means that
the trajectory starting at $\xx_0$ leaves $A$ immediately and keep outside $A$ for some time in future;  $\xx_0\in \mbox{IvOut}_{\fb}(A)$ means that
the trajectory passing through $\xx_0$ reaches $\xx_0$ from the outside of $A$.

Based on the above notations, we have

\begin{theorem}\label{thm:topo}
Given a CDSwDI $(H,\fb, \Xi)$ with $H\subseteq \RR^n$, $\Xi\subseteq \RR^n$ and locally Lipschitz continuous $\fb:\RR^n\mapsto\RR^n$, a subset $P$ of $\RR^n$ is a CI of $(H,\fb, \Xi)$ if and only if
\begin{enumerate}
  \item $\Xi \subseteq P$;
  \item $\forall\xx\in P\cap H\cap \,\emph{In}_{\fb}(H).\,\xx\in \emph{In}_{\fb}(P)$;
  \item $\forall\xx\in P^c\cap H\cap\, \emph{IvIn}_{\fb}(H).\,\xx\in \big(\emph{IvIn}_{\fb}(P)\big)^c$.
\end{enumerate}
\end{theorem}

%

\begin{proof}

First of all, the proof about condition 1 is trivial. In what follows, we focus on the proofs about conditions 2 and 3.

``$\Leftarrow$" Suppose $P$ is not a CI of $(H,\fb, \Xi)$. According to Definition \ref{dfn:inv}, there exists $\xx_0\in P\cap H$, $T_0> 0$ and $T_1\in(0,T_0]$ s.t.
$$\forall t\in[0,T_0].\,\xx(t)\in H \quad\mbox{and}\quad \xx(T_1)\notin P.$$
It is not difficult to check that the set
$$\mathcal T_P\,\define\,\{T\in\RR,T\geq 0\mid\forall t \in [0,T].\,\xx(t)\in P\}$$
is not empty, and is a right-open or right-closed interval $[0,T_P\rg$ with
$0\leq T_P\leq T_1$. If $[0,T_P\rg=[0,T_P]$, then $T_P<T_1$. Thus $\xx(T_P)\in P\cap H\cap\, \mbox{In}_{\fb}(H)$, but $\xx(T_P)\notin \mbox{In}_{\fb}(P)$, otherwise $T_P$ could not be the right end point of $\mathcal T_P$. So 2 is violated.

If $[0,T_P\rg=[0,T_P)$, then $T_P>0$ and $\xx (T_P)\in P^c\cap H$. Furthermore, $\forall t\in[0,T_P).\, \xx(t)\in P\cap H$, i.e.
$$\forall t\in[0,T_P).\, \xx(\xx_0;t)\in P\cap H,$$
which is equivalent to
$$\forall t\in[-T_P,0).\, \xx(\xx_0;t+T_P)\in P\cap H.$$
Let $\xx_0'=\xx(\xx_0;T_P)$. Then $\xx(\xx_0';t)=\xx(\xx_0;t+T_p)$. Thus we get
$$\forall t\in[-T_P,0).\, \xx(\xx_0';t)\in P\cap H,$$
i.e.
$$\forall t\in(0,T_P].\, \xx(\xx_0';-t)\in P\cap H.$$
This means $\xx_0'\in \mbox{IvIn}_{\fb}(H)\cap \mbox{IvIn}_{\fb}(P).$
Besides,
$$\xx_0'=\xx(\xx_0';0)=\xx(\xx_0;T_P)=\xx(T_P)\in P^c\cap H.$$
So 3 is violated by $\xx_0'$.

``$\Rightarrow$" If 2 does not hold,
then there exists $\xx_1\in P\cap H$, $\epsilon_1>0$  and $0<t_1<\epsilon_1$ such that $\forall t\in [0,\epsilon_1).\,\xx(\xx_1;t)\in H$ and $\xx(\xx_1;t_1)\notin P$. By Definition \ref{dfn:inv}, $P$ could not be a CI.

If 3 does not hold, then there exists
$$\xx_2\in P^c\cap H\cap\mbox{IvIn}_{\fb}(H)\cap \mbox{IvIn}_{\fb}(P).$$
This means there exists $\epsilon_2>0$ such that
$$\forall t\in (0,\epsilon_2).\,\xx(\xx_2;-t)\in P\cap H,$$
i.e.
$$\forall t\in (-\epsilon_2,0).\,\xx(\xx_2;t)\in P\cap H.$$
Thus
$$\forall t\in [-\epsilon_2/2,0).\,\xx(\xx_2;t)\in P\cap H.$$
i.e.
$$\forall t\in [0,\epsilon_2/2).\,\xx(\xx_2;t-\epsilon_2/2)\in P\cap H.$$
Let $\xx_2'=\xx(\xx_2;-\epsilon_2/2)$. Then $\xx(\xx_2';t)=\xx(\xx_2;t-\epsilon_2/2)$. Thus we get
$$\forall t\in [0,\epsilon_2/2).\,\xx(\xx_2';t)\in P\cap H.$$
Furthermore, $$\xx(\xx_2';\epsilon_2/2)=\xx(\xx_2;0)=\xx_2\in P^c\cap H.$$
Thus the trajectory starting from $\xx_2'$ violates the condition of Definition \ref{dfn:inv}, so $P$ could not be a CI either.
\end{proof}

\subsection{Necessary and Sufficient Condition for SAI}

Given a PDS $(H,\fb, \Xi)$ and an SAI $P$, to encode the conditions in Theorem \ref{thm:topo} as polynomial formulas, it is sufficient to show that $\mbox{In}_{\fb}(H)$, $\mbox{In}_{\fb}(P)$, $\mbox{IvIn}_{\fb}(H)$ and $\mbox{IvIn}_{\fb}(P)$ are all semi-algebraic sets. By the structure of $H$, it is natural to consider the relation between $\mbox{In}_{\fb}(H)$ and $\mbox{In}_{\fb}\big(\mathcal S(p_{ij}\triangleright 0)\big)$. Through a careful analysis, we establish the following crucial equality:

\begin{theorem} \label{thm:dcp-inward} For a semi-algebraic set $H$ defined by formula (\ref{for:dom}) and a polynomial vector field $\fb$, we have
$$\emph{In}_{\fb}(H)=\bigcup_{i=1}^I\bigcap_{j=1}^{J_i} \emph{In}_{\fb}\big(\mathcal S(p_{ij}\triangleright 0)\big).$$
\end{theorem}

To prove Theorem \ref{thm:dcp-inward}, we need the following two Lemmas, wherein $\triangleright\in \{\geq,>\}$.

\begin{lemma}\label{lm:atm}
For any atomic polynomial formula $p\triangleright 0$ and polynomial vector field $\fb$, and for any $\xx_0\in\RR^n$, we have either $\xx_0\in \emph{In}_{\fb}\big(\mathcal S(p\triangleright 0)\big)$ or $\xx_0\in \emph{Out}_{\fb}\big(\mathcal S(p\triangleright 0)\big)$.
\end{lemma}
\begin{proof}
Polynomial functions are analytic, so $\fb$ is analytic and thus
$\xx(\xx_0;t)$ ($\xx(t)$ for short) is analytic in a small interval $(a,b)$ containing $0$. Besides, $p$ is analytic, so the Taylor expansion of $p(\xx(t))$ at $t=0$
\begin{eqnarray}
p(\xx(t))& = &p(\xx_0)+\frac{\ud p}{\ud t}\cdot t + \frac{\ud^{2} p}{\ud t^{2}}\cdot\frac{t^{2}}{2!}+\cdots \nonumber\\
& = & L_{\fb}^0 p(\xx_0)+L_{\fb}^1 p(\xx_0)\cdot t + L_{\fb}^2
p(\xx_0)\cdot \frac{t^{2}}{2!}+\cdots\quad \nonumber
\end{eqnarray}
converges in $(a,b)$. Then the proof proceeds by case analysis on the sign of
$L_{\fb}^{\gamma_{p,\fb}(\xx_0)}p(\xx_0)$:
\begin{itemize}
  \item if $\gamma_{p,\fb}(\xx_0)=\infty$, then $\exists \epsilon>0\forall t\in (0,\epsilon). \,p(\xx(t))=0$, so $\xx_0\in \mbox{In}_{\fb}\big(\mathcal S(p\geq 0)\big)$ and $\xx_0\in \mbox{Out}_{\fb}\big(\mathcal S(p> 0)\big)$;
  \item if $L_{\fb}^{\gamma_{p,\fb}(\xx_0)}p(\xx_0)>0$, then $\exists \epsilon>0\forall t\in (0,\epsilon). \,p(\xx(t))>0$, so $\xx_0\in \mbox{In}_{\fb}\big(\mathcal S(p\geq 0)\big)$ and $\xx_0\in \mbox{In}_{\fb}\big(\mathcal S(p> 0)\big)$;
  \item if $L_{\fb}^{\gamma_{p,\fb}(\xx_0)}p(\xx_0)<0$, then $\exists \epsilon>0\forall t\in (0,\epsilon). \,p(\xx(t))<0$, so $\xx_0\in \mbox{Out}_{\fb}\big(\mathcal S(p\geq 0)\big)$ and $\xx_0\in \mbox{Out}_{\fb}\big(\mathcal S(p> 0)\big)$.
\end{itemize}
Then we can see that for all $\xx_0\in\RR^n$, either $\xx_0\in \mbox{In}_{\fb}\big(\mathcal S(p\,\triangleright 0)\big)$ or $\xx_0\in \mbox{Out}_{\fb}\big(\mathcal S(p\triangleright 0)\big)$.
\end{proof}

\begin{lemma}\label{lm:conjunct}
For any semi-algebraic set $B=\mathcal S\big(\bigwedge_{j=1}^J p_j\triangleright 0\big)$, and polynomial vector field, we have
\begin{enumerate}
  \item $\emph{In}_{\fb}(B)=\bigcap_{j=1}^J\emph{In}_{\fb}\big(\mathcal S(p_j\triangleright 0)\big)$;
  \item for any $\xx_0\in\RR^n$, either $\xx_0\in\emph{In}_{\fb}(B)$ or $\xx_0\in\emph{Out}_{\fb}(B)$.
\end{enumerate}
\end{lemma}
\begin{proof}
\begin{enumerate}
  \item ``$\subseteq$" Trivial.

  ``$\supseteq$" For any $\xx_0\in \bigcap_{j=1}^J\mbox{In}_{\fb}\big(\mathcal S(p_j\triangleright 0)\big)$, there exist positive $\epsilon_1,\epsilon_2,\ldots,\epsilon_J$ such that for all $1\leq j\leq J$ and any $t\in (0,\epsilon_j)$, $p_j(\xx(\xx_0;t))\triangleright 0$. Let $\epsilon=\min\{\epsilon_1,\epsilon_2,\ldots,\epsilon_J\}$. Then for any
  $t\in (0,\epsilon)$, $\bigwedge_{j=1}^J p_j(\xx(\xx_0;t))\triangleright 0$. Thus $\xx_0\in\mbox{In}_{\fb}(B)$.
  \item By 1 if $\xx_0\notin \mbox{In}_{\fb}(B)$, then there exists $j_0\in[1,J]$ such that $\xx_0\notin \mbox{In}_{\fb}\big(\mathcal S(p_{j_0}\triangleright 0)\big)$. By Lemma \ref{lm:atm}, $\xx_0\in \mbox{Out}_{\fb}\big(\mathcal S(p_{j_0}\triangleright 0)\big)$. Thus there exists $\epsilon>0$ s.t. for all $t\in(0,\epsilon)$, $\neg\big(p_{j_0}(\xx(\xx_0;t))\triangleright 0\big)$. Then for all $t\in(0,\epsilon)$, $\bigvee_{j=1}^J \neg\big(p_{j}(\xx(\xx_0;t))\triangleright 0\big)$, i.e. $\neg\big(\bigwedge_{j=1}^J p_j(\xx(\xx_0;t))\triangleright 0\big)$.
      This means $\xx_0\in\mbox{Out}_{\fb}(B)$.
\end{enumerate}
\end{proof}
Now we are ready to prove Theorem \ref{thm:dcp-inward} as follows.
\begin{proof}[of  Theorem \ref{thm:dcp-inward}]
``$\supseteq$" Trivial.

``$\subseteq$"\, If $\xx_0\notin \bigcup_{i=1}^I\bigcap_{j=1}^J \mbox{In}_{\fb}\big(\mathcal S(p_{ij}\triangleright 0)\big)$, then for all $i\in [1,I]$, $\xx_0 \notin \bigcap_{j=1}^J \mbox{In}_{\fb}\big(\mathcal S(p_{ij}\triangleright 0)\big)$. By Lemma \ref{lm:conjunct}, for all $i\in [1, I]$, $\xx_0\in \mbox{Out}_{\fb}(B)$, where $B=\bigwedge_{j=1}^J p_{ij}\triangleright 0$. Thus there exist positive $\epsilon_1,\epsilon_2,\ldots,\epsilon_I$ s.t.
for all $i\in [1,I]$ and any $t\in (0,\epsilon_i)$, $\neg\big(\bigwedge_{j=1}^J p_{ij}(\xx(\xx_0;t))\triangleright 0\big)$. Let $\epsilon=\min\{\epsilon_1,\epsilon_2,\ldots,\epsilon_I\}$. Then for all $t\in (0,\epsilon)$, $\bigwedge_{i=1}^I \neg\big(\bigwedge_{j=1}^J p_{ij}(\xx(\xx_0;t))\triangleright 0\big)$, or equivalently, $\neg\big(\bigvee_{i=1}^I\bigwedge_{j=1}^J p_{ij}(\xx(\xx_0;t))\triangleright 0\big)$. This means $\xx_0\in\mbox{Out}_{\fb}(H)$ and $\xx_0\notin \mbox{In}_{\fb}(H)$.
\end{proof}

Based on Theorem \ref{thm:dcp-inward}, in order to show $\mbox{In}_{\fb}(H)$ is a semi-algebraic set  for any semi-algebraic set $H$, it is sufficient to show that $\mbox{In}_{\fb}\big(\mathcal S(p\triangleright 0)\big)$ is a semi-algebraic set  for any polynomial $p$, where $\triangleright \in\{\geq,>\}$.

In fact, we have proved in Lemma \ref{lm:atm} the following result.

\begin{lemma}\label{lm:Gamma+0}
  For any polynomial $p$ and polynomial vector field $\fb$,
  \begin{align*}
  \emph{In}_{\fb}(\mathcal S(p> 0))&=\Gamma_{+}(p,\fb)\quad\quad\,\, \mbox{and}\\
  \emph{In}_{\fb}(\mathcal S(p\geq 0))&=\Gamma_{0}(p,\fb)\cup\Gamma_{+}(p,\fb)\enspace ,
  \end{align*}
  where
  \begin{align}
  \Gamma_{0}(p,\fb)&\,\define\, \{\xx_0\in\RR^n\mid \gamma_{p,\fb}(\xx_0)=\infty\}\,\,\mbox{and}\\
  \Gamma_+(p,\fb) &\,\define\, \{\xx_0\in\RR^n\mid \gamma_{p,\fb}(\xx_0)<0 \,\wedge\, L_{\fb}^{\gamma_{p,\fb}(\xx_0)}p(\xx_0)>0\}.
  \end{align}
\end{lemma}
Next, we show $\Gamma_{0}$ and $\Gamma_{+}$ are semi-algebraic sets. We will do so in a more general way for parametric polynomials $p(\uu,\xx)$. In their proofs, we need the fundamental results about Lie derivatives shown in Section \ref{sec:fund}. In the sequel we adopt the convention that $\bigwedge_{i\in \emptyset}\phi_i=true$, where $\phi_i$ is a polynomial formula.

\begin{lemma} \label{lm:Gamma-0}
Given $p\,\define\,p(\uu,\xx)$ and polynomial vector field $\fb$, for any $\uu_0\in\RR^t$ we have
$$\Gamma_0(p_{\uu_0},\fb)=\mathcal S\big(\varphi_0(p,\fb)\mid_{\uu=\uu_0}\big)\enspace,$$
where
\begin{equation}\label{eqn:phi0}\varphi_{0}(p,\fb)\,\define\, \bigwedge_{i=0}^{N_{p,\fb}}L_{\fb}^i p =0\enspace.\end{equation}
\end{lemma}
\begin{proof}
``$\subseteq$" This is trivial by definition of pointwise rank in Section \ref{sec:prelim}.

``$\supseteq$" If $\xx_0\in \mathcal S\big(\varphi_0(p,\fb)\mid_{\uu=\uu_0}\big)$, then by definition of pointwise rank we have $\gamma_{p_{\uu_0},\fb}(\xx_0)>N_{p,\fb}$. By the similarity of Theorem \ref{Rank} with parameters in polynomial $p$, we get $\gamma_{p_{\uu_0},\fb}(\xx_0)=\infty$. Thus $\xx_0\in \Gamma_0(p_{\uu_0},\fb).$
\end{proof}

\begin{lemma}\label{lm:Gamma+}
Given $p\,\define\,p(\uu,\xx)$ and polynomial vector field $\fb$, for any $\uu_0\in\RR^t$ we have
$$\Gamma_+(p_{\uu_0},\fb)=\mathcal S\big(\psi_+(p,\fb)\mid_{\uu=\uu_0}\big)\enspace,$$
where
\begin{equation}\label{eqn:psi+}\psi_{+}(p,\fb)\,\define\, \bigvee_{i=0}^{N_{p,\fb}} \psi^{(i)} (p,\fb)\qquad \quad\mbox{with}\end{equation}
$$\psi^{(i)}(p,\fb)\,\define\,\Big(\bigwedge_{j=0}^{i-1}L_{\fb}^j p=0 \Big)\wedge L_{\fb}^i p>0\enspace .$$
\end{lemma}
\begin{proof}
``$\supseteq$" If $\xx_0\in \mathcal S\big(\varphi_+(p,\fb)\mid_{\uu=\uu_0}\big)$, then by definition of pointwise rank, we have $$\Big(\gamma_{p_{\uu_0},\fb}(\xx_0)\leq N_{p,\fb}<\infty \Big)\,\wedge\, L_{\fb}^{\gamma_{p_{\uu_0},\fb}(\xx_0)}p_{\uu_0}(\xx_0)>0\,.$$
Thus $\xx_0\in \Gamma_+(p_{\uu_0},\fb)$\enspace .

``$\subseteq$" If $\xx_0\in \Gamma_+(p_{\uu_0},\fb)$, then by definition of pointwise rank we know  $\xx_0$ satisfies
$$L_{\fb}^0 p_{\uu_0}=0\wedge \cdots\wedge L_{\fb}^{\scriptscriptstyle{\gamma_{p_{\uu_0},\fb}(\xx_0)-1}}p_{\uu_0}=0
\wedge L_{\fb}^{\scriptscriptstyle{\gamma_{p_{\uu_0},\fb}(\xx_0)}}p_{\uu_0}>0\enspace . $$
 By the similarity of Theorem \ref{Rank} with parameters in polynomial $p$, we have $\gamma_{p_{\uu_0},\fb}(\xx_0)\leq N_{p,\fb}$. Thus $\uu_0,\xx_0$ satisfy $\phi^{\gamma_{p_{\uu_0},\fb}}(p,\fb)$. This means  $\xx_0\in \mathcal S\big(\varphi_+(p,\fb)\mid_{\uu=\uu_0}\big)$.
\end{proof}

Based on Lemma \ref{lm:Gamma+0}, \ref{lm:Gamma-0} and \ref{lm:Gamma+} we have
\begin{theorem}\label{lm:in-atm}
  For any polynomial $p$ and vector field $\fb$,
  \begin{align*}
  \emph{In}_{\fb}(\mathcal S(p> 0))&=\mathcal S(\psi_+(p,\fb)), \text{ and}\\
  \emph{In}_{\fb}(\mathcal S ( p\geq 0))&=\mathcal S\big(\psi_+(p,\fb) \vee \varphi_0(p,\fb)\big)
  \end{align*}
  where $\varphi_0(p,\fb)$ and $\psi_+(p,\fb)$ are defined in (\ref{eqn:phi0}) and (\ref{eqn:psi+}) respectively.
\end{theorem}

Therefore,  $\mbox{In}_{\fb}(H)$ can be translated into a polynomial formula. By a similar argument, we are able to prove that

\begin{theorem} \label{coro:dcp-ivin}  For a semi-algebraic set $H$ defined by formula (\ref{for:dom}) and a polynomial vector field $\fb$, we have
$$\emph{IvIn}_{\fb}(H)=\bigcup_{i=1}^I\bigcap_{j=1}^{J_{i}} \emph{IvIn}_{\fb}\big(\mathcal S(p_{ij}\triangleright 0)\big).$$
\end{theorem}

Accordingly,
\begin{theorem}\label{coro:semi-ivin-atm}
For any polynomial $p$ and vector field $\fb$,
\begin{align*}
\emph{IvIn}_{\fb}\big(\mathcal S(p>0)\big)&=\mathcal S\big(\varphi_+(p,\fb)\big), \text{ and} \\
   \emph{IvIn}_{\fb}\big(\mathcal S(p\geq 0)\big)&=\mathcal S\Big(\varphi_+(p,\fb)\vee \varphi_0(p,\fb)\Big)
\end{align*}

where
\begin{align}
  \varphi_{+}(p,\fb)&\,\define\, \bigvee_{i=0}^{N_{p,\fb}} \varphi^{(i)} (p,\fb)\qquad \mbox{with}\label{eqn:phi-iv+}\\
 \varphi^{(i)}(p,\fb) &\,\define\, \Big(\bigwedge_{j=0}^{i-1}L_{\fb}^j p=0 \Big)\wedge \Big((-1)^i\cdot L_{\fb}^i p>0\Big).\nonumber
\end{align}
\end{theorem}

Now we are able to present our main result of automatic SAI generation for PDS.
\begin{theorem}[Main Result]\label{thm:main}
A semi-algebraic set \\$\mathcal S(P)$ with
$$P\,\define \, \bigvee_{k=1}^{K} \left( \bigwedge_{j=1}^{j_k} p_{kj}(\uu_{kj},\xx)\geq 0 \quad\wedge\, \bigwedge_{j=j_k+1}^{J_k} p_{kj}(\uu_{kj},\xx)>0 \right)
$$
is a continuous invariant of the PDS $\big(\mathcal S(H),\fb, \Xi\big)$ with
$$H\,\define\, \bigvee_{m=1}^{M}\left( \bigwedge_{l=1}^{l_m} p_{ml}(\xx)\geq 0 \quad \wedge \, \bigwedge_{l=l_m+1}^{L_m} p_{ml}(\xx)>0 \right),$$
if and only if
$\uu\define \langle \uu_{kj} \rangle$ satisfy
$$\forall \xx. \left(\begin{array}{l}
    (\Xi(\xx)\rightarrow P(\uu,\xx)) \wedge \\
    \big(P\wedge H\wedge \varphi_{H}\rightarrow \varphi_{P}\big) \wedge \big(\neg P \wedge H\wedge \varphi_H^{\scriptsize{\emph{Iv}}} \rightarrow \neg \varphi_P^{\scriptsize{\emph{Iv}}}\,\big)
    \end{array} \right),$$
where
\begin{align*}
  \varphi_H &\,\define\, \bigvee_{m=1}^{M}\left( \bigwedge_{l=1}^{l_m} \psi_{0,+}(p_{ml},\fb) \wedge \, \bigwedge_{l=l_m+1}^{L_m} \psi_+(p_{ml},\fb) \right),\\
  \varphi_P &\,\define\,\bigvee_{k=1}^{K}\left( \bigwedge_{j=1}^{j_k} \psi_{0,+}(p_{kj},\fb) \wedge \, \bigwedge_{j=j_k+1}^{J_k} \psi_+(p_{kj},\fb) \right),\\
  \varphi_H^{\scriptsize{\emph{Iv}}} &\,\define\,\bigvee_{m=1}^{M}\left( \bigwedge_{l=1}^{l_m} \varphi_{0,+}(p_{ml},\fb) \wedge \, \bigwedge_{l=l_m+1}^{L_m} \varphi_+(p_{ml},\fb) \right),\\
  \varphi_P^{\scriptsize{\emph{Iv}}} &\,\define\,\bigvee_{k=1}^{K}\left( \bigwedge_{j=1}^{j_k} \varphi_{0,+}(p_{kj},\fb) \wedge \, \bigwedge_{j=j_k+1}^{J_k} \varphi_+(p_{kj},\fb) \right),
\end{align*}
with $\psi_{0,+}(p,\fb)\,\define\,\psi_+(p,\fb)\vee\varphi_0(p,\fb)$ and 
$\varphi_{0,+}(p,\fb)\,\define\,\varphi_+(p,\fb)\vee\varphi_0(p,\fb)$.
\end{theorem}
\begin{proof}
This theorem is a direct consequence of Theorem \ref{thm:topo}, \ref{thm:dcp-inward}, \ref{lm:in-atm}, \ref{coro:dcp-ivin} and \ref{coro:semi-ivin-atm}.
\end{proof}

Note that $\varphi_H$ and $\varphi_H^{\text{Iv}}$ are trivially ``true" when $H$ is the whole space $\mathbb R^n$.

Compared to related work, e.g \cite{PlatzerClarke08,Prajna04,Prajna07,Sankar10}, our method for SAI generation based on Theorem \ref{thm:main} has the following two features:
\begin{enumerate}
  \item Given a PDS (with arbitrary semi-algebraic domain and initial states), we consider arbitrary semi-algebraic sets as invariants, which are of complicated forms and may be neither open nor closed.
  \item Our criterion for checking semi-algebraic invariants for PDS is sound and complete; our method for automatically generating semi-algebraic invariants is sound, and complete w.r.t to the predefined template.
\end{enumerate}

Now we demonstrate how our approach can be used to generate a general SAI by the following example.
\begin{example}
Let  $\fb(x,y)=(\dot{x}=-2y, \dot{y}=x^2)$ with $H\, \define \, \mathbb R^2$ and $\Xi \define x+y\geq 0$. Take a template:  $\tau\define x-a\geq 0 \vee y-b> 0$. By Theorem \ref{thm:main}, $\tau$ is an SAI of $(H,\fb,\Xi)$ iff $(a,b)$ satisfies the following two formulas
\begin{align}
~& x+y\geq 0 \rightarrow (x-a\geq 0 \vee y-b> 0) \\
~& (\tau  \rightarrow \zeta) \wedge (\neg \tau \rightarrow \neg \xi)
\end{align}
for all $(x,y)\in \RR ^2$, where
\begin{align*}
  \zeta  \define &  (x-a>0) \vee (x-a=0 \wedge  -2y>0)  \\
  & \vee (x-a=0 \wedge -2y=0 \wedge -2x^2\geq 0) \\
  &\vee (y-b>0) \vee (y-b=0 \wedge x^2>0)  \\
  &\vee (y-b=0 \wedge x^2=0 \wedge -4yx>0)  \\
  & \vee(y-b=0 \wedge x^2=0 \wedge -4yx=0 \wedge 8y^2-4x^3>0)\\
  \xi \define & (x-a>0) \vee (x-a=0 \wedge  -2y<0)  \\
  & \vee(x-a=0 \wedge  -2y=0 \wedge -2x^2\geq 0) \\
  & \vee (y-b>0) \vee (y-b=0 \wedge x^2<0) \\
  & \vee(y-b=0 \wedge x^2=0 \wedge -4yx>0) \\
  & \vee (y-b=0 \wedge x^2=0 \wedge -4yx=0 \wedge 8y^2-4x^3<0)
\end{align*}

By applying quantifier elimination to this formula, we get $a+b\leq 0 \wedge b\leq 0$. Let  $a=-1$ and $b=-0.5$, and it results that $\{(x,y)\in \mathbb R^2 \mid x\geq -1 \vee y> -0.5\}$ is an SAI for this PDS, which is shown in IV of Figure \ref{fig:Inv}.
\end{example}

Note that in the above example, the generated SAI is a general semi-algebraic set that is a union of two simple semi-algebraic sets, which is neither closed nor open.

\section{Case Study}\label{sec:case}

In this section, we show that our method presented above can be used to generate continuous invariants for some real systems.

\subsection{Formal Verification of CTCS-3}

In \cite{aplas},  the authors  use \emph{HCSP} \cite{He94, Zhou95} to formally model the \emph{Chinese Train Control System at Level 3 (CTCS-3)} \cite{Zhang08}. They also propose a calculus of HCSP  for the purpose of
verifying safety properties of CTCS-3. For this calculus to work, effective techniques for dealing with continuous dynamics must be incorporated.

Consider the following fragment of the HCSP model of CTCS-3:
$$
P_{\textit{ebi}}\,\,\define\,\, \la \dot{s}=v, \dot{v}=a\ra
\rightarrow v\geq v.\textit{Seg} \,; \textit{flag}_\textit{EB}:=
\textit{true}\,; P_\textit{EB}\,.
$$
Process $P_{\textit{ebi}}$ models the running of a train, with $s,v,a$ representing its position, velocity and acceleration ($a$ is a constant) respectively.
Once $v$ exceeds the speed limit $v.\textit{Seg}$ of the current segment, $\textit{flag}_\textit{EB}$ for emergency brake is set to \textit{true} and the train
starts braking immediately, expressed by the subprocess
$P_\textit{EB}$.

The safety property needs to be verified about $P_\textit{ebi}$ can
be stated as
$$ \textit{Inv}\,\, \define\, v \geq v.\textit{Seg}\rightarrow
\textit{flag}_{\textit{EB}}=\textit{true}\,,$$ which means whenever the train's speed exceeds certain limit, it must execute the emergency brake process.

To verify this property, i.e. to check that $\textit{Inv}$ is indeed an invariant of $P_{\textit{ebi}}$, according to the calculus in \cite{aplas}, it amounts to check
that $v < v.\textit{Seg}$ is a continuous invariant of the PDS $(H,\fb,\Xi)$, where
$H\,\define\,\mathcal S(v < v.\textit{Seg})$, $\fb\,\define\,(v,a)$ and $\Xi\,\define\,\{(s_0,v_0)\}$ with $v_0<v.\textit{Seg}$. According to our method, this can be further reduced to the checking of
the validity of
$$ \forall v. (v=v.\textit{Seg}\wedge
v<v.\textit{Seg}\rightarrow a\leq 0),$$ which is obvious.

Perhaps this example seems a bit trivial, for the continuous dynamics is an affine system and the required invariant coincides with the domain. What we want to stress here is the completeness of our criterion for checking continuous invariants compared to others. For example, the principle given in \cite{PlatzerClarke08} requires the directional derivative of an invariant in the direction of the vector field to have the same sign in the domain. As a result, it may fail to generate the above invariant $\mathcal S(v < v.\textit{Seg})$, because
$$\forall v.(v < v.\textit{Seg}\rightarrow \dot{v}=a <
0)$$
is $\textit{false}$ when $a\geq 0$.

\subsection{Collision Avoidance Maneuvers}
We consider the following two-aircraft flight dynamics from \cite{PlatzerClarke09}:
\begin{equation}\label{eqn:flight}
\fb\,\define\,
  \left[\begin{array}{llll}
  \dot x_1=d_1 & \dot y_1=e_1 & \dot d_1=-\omega d_2 & \dot e_1= -\theta e_2 \\
\dot x_2=d_2 & \dot y_2=e_2 & \dot d_2=\omega d_1 & \dot e_2= \theta e_1 \\
  \end{array}\right]\,.
\end{equation}
System (\ref{eqn:flight}) has 8 variables: $(x_1,x_2)$ and $(y_1,y_2)$ represent the positions of aircraft 1 and 2 respectively, and $(d_1,d_2)$ and $(e_1,e_2)$ represent their velocities.  The  parameters $\omega$ and $\theta$ denote the angular speed of the two aircrafts.

We shall apply our method to generating special invariants of form $p=0$ for PDS $(H,\fb,\Xi)$ with $H\,\define\,\mathbb R^8$ and $\fb$ defined in (\ref{eqn:flight}). For simplicity, we take $\Xi$ to be a singleton
$\{(x_1^0,x_2^0,d_1^0,d_2^0,y_1^0,y_2^0,e_1^0,e_2^0)\}$\,.

In order to determine candidates for invariants of $(H,\fb,\Xi)$, we enumerate parametric polynomials $p\,\define\,p(\uu,\xx)$ by the degree of $p$ and the number of variables appearing in it. For example, we can choose the linear template $p(\uu,\xx)\,\define\,u_1 x_1 + u_2 x_2 + u_3 d_1 + u_4 d_2 + u_0$.

According to Theorem \ref{thm:main}, it is easy to check that $p(\uu,\xx)=0$ is an invariant of $(H,\fb,\Xi)$ if and only if $\uu$ satisfies
\begin{itemize}
  \item $\forall \xx. \,\Xi\rightarrow p=0$\,;\, and
  \item $\forall \xx. \, p=0\rightarrow \bigwedge_{i=1}^{N_{p,\fb}}L_{\fb}^i p(\uu,\xx)=0$\,.
\end{itemize}
For the template defined above,  we can get $N_{p,\fb}=2$. By applying quantifier elimination to the corresponding constraint, we get
$\,u_2 - u_3 \omega = 0 \wedge u_1 + u_4 \omega= 0 \wedge
u_0 + u_1 x_1^0 + u_2 x_2^0 + u_3 d_1^0  + u_4 d_2^0 = 0\,.$
Thus we can obtain the following invariants by assigning suitable values to $u_i$s:
\begin{itemize}
  \item $\omega x_2 + d_1 - \omega x_2^0 - d_1^0=0$;
  \item $-\omega x_1 + d_2 + \omega x_1^0 - d_2^0 =0$;
  \item $-\omega x_1 + \omega x_2 + d_1 + d_2 + \omega x_1^0- \omega x_2^0 - d_1^0  - d_2^0 =0$.
\end{itemize}

If we use the quadratic template $p\,\define\, u_1 d_1^2 + u_2 d_2^2 +u_0$, we can also get $N_{p,\fb}=2$, and the constraint for $\uu$ is
$\,u_1 - u_2 = 0 \wedge  u_0  + u_1 (d_1^0)^2 + u_2 (d_2^0)^2 = 0\,.$
Let $u_1=u_2=1$ and we obtain an invariant
$$d_1^2+d_2^2- (d_1^0)^2 - (d_2^0)^2=0\enspace .$$

Using arbitrary semi-algebraic templates, we can generate invariants beyond polynomial equations for $(H,\fb,\Xi)$, at the cost of heavier computation.

\section{Conclusions}\label{sec:con}
In this paper, we present a sound and complete criterion for checking SAIs for PDSs, as well as a relatively complete method for automatic SAI generation using templates. Our approach is based on the computable algebraic-geometry theory.  Our work in this paper actually completes the gap left open in \cite{Taly09}. Compared with the related work, more invariants can be generated through our approach. This is demonstrated by simple examples and case studies.

In the future, we will concentrate on the following problems. Firstly, we believe that our method can be applied to generate invariance sets for stability analysis, controller synthesis and so on  in control theory, in particular for construction of Lyapunov functions. Secondly, we will consider how to extend the approach to more general dynamical systems whose vector fields are functions beyond polynomials. Since our approach makes use of first-order quantifier elimination which is with doubly exponential cost \cite{dh88}, how to improve the efficiency of our approach will be our main future work. For instance of linear templates, it is helpful to reduce the complexity via linear programming.

%
%

\begin{thebibliography}{10}

\bibitem{Alur94}
R.~Alur and D.~L. Dill.
\newblock A theory of timed automata.
\newblock {\em Theor. Comput. Sci.}, 126(2):183--235, 1994.

\bibitem{Alur95}
R.~Alur and et~al.
\newblock The algorithmic analysis of hybrid systems.
\newblock {\em Theor. Comput. Sci.}, 138(1):3--34, 1995.

\bibitem{Blanchini}
F.~Blanchini.
\newblock Set invariance in control.
\newblock {\em Automatica}, 35(11):1747--1767, 1999.

\bibitem{qepcad}
C.~W. Brown.
\newblock {QEPCAD B}: A program for computing with semi-algebraic sets using
  {CAD}s.
\newblock {\em SIGSAM Bulletin}, 37:97--108, 2003.

\bibitem{CE81}
E.~M. Clarke, E.~A. Emerson, and A.~P. Sistla.
\newblock Automatic verification of finite-state concurrent systems using
  temporal logic specifications.
\newblock {\em ACM Trans. Program. Lang. Syst.}, 8(2):244--263, 1986.

\bibitem{clo}
D.~Cox, J.~Little, and D.O'Shea.
\newblock {\em Ideals, Varieties, and Algorithms: An Introduction to
  Computational Algebraic Geometry and Commutative Algebra}.
\newblock Springer, 1996.

\bibitem{dh88}
J.~H. Davenport and J.~Heintz.
\newblock Real quantifier elimination is doubly exponential.
\newblock {\em J. Symb. Comput.}, 5(1/2):29--35, 1988.

\bibitem{Redlog}
A.~Dolzmann and T.~Sturm.
\newblock Redlog user manual - edition 2.0, for redlog version 2.0.
\newblock 1999.

\bibitem{Tiwari08}
S.~Gulwani and A.~Tiwari.
\newblock Constraint-based approach for analysis of hybrid systems.
\newblock {\em {\upshape In} CAV'08, LNCS}, 5123:190--203, 2008.

\bibitem{Haddad-Che}
W.~M. Haddad and V.~Chellaboina.
\newblock {\em Nonlinear Dynamical Systems and Control: A {Lyapunov}-Based
  Approach}.
\newblock Princeton University Press, 2008.

\bibitem{He94}
J.~He.
\newblock From {CSP} to hybrid systems.
\newblock {\em {\upshape In} A Classical Mind: Essays in Honour of C. A. R.
  Hoare, Prentice-Hall International Series In Computer Science}, pages
  171--189, 1994.

\bibitem{Henzinger95}
T.~A. Henzinger and et~al.
\newblock What's decidable about hybrid automata?
\newblock {\em {\upshape In} STOC'95}, pages 373--382, 1995.

\bibitem{Krantz}
S.~Krantz and H.~Parks.
\newblock {\em A Primer of Real Analytic Functions}.
\newblock Birkh{\"a}user Boston, second edition, June 2002.

\bibitem{LPY02}
G.~Lafferriere, G.~J. Pappas, and S.~Yovine.
\newblock Symbolic reachability computation for families of linear vector
  fields.
\newblock {\em J. Symb. Comput.}, 32(3), 2001.

\bibitem{aplas}
J.~Liu, J.~Lv, Z.~Quan, N.~Zhan, H.~Zhao, C.~Zhou, and L.~Zou.
\newblock A calculus for hybrid {CSP}.
\newblock {\em {\upshape In} APLAS'10, LNCS}, 6461:1--15, 2010.

\bibitem{Platzer10}
A.~Platzer.
\newblock Differential-algebraic dynamic logic for differential-algebraic
  programs.
\newblock {\em J. Log. Comput.}, 20(1):309--352, 2010.

\bibitem{PlatzerClarke08}
A.~Platzer and E.~M. Clarke.
\newblock Computing differential invariants of hybrid systems as fixedpoints.
\newblock {\em Form. Methods Syst. Des.}, 35(1):98--120, 2009.

\bibitem{PlatzerClarke09}
A.~Platzer and E.~M. Clarke.
\newblock Formal verification of curved flight collision avoidance maneuvers: A
  case study.
\newblock {\em {\upshape In} FM '09, LNCS}, 5850:547--562, 2009.

\bibitem{Prajna04}
S.~Prajna and A.~Jadbabaie.
\newblock Safety verification of hybrid systems using barrier certificates.
\newblock {\em {\upshape In} HSCC'04, LNCS}, 2993:477--492, 2004.

\bibitem{Prajna07}
S.~Prajna, A.~Jadbabaie, and G.~J. Pappas.
\newblock A framework for worst-case and stochastic safety verification using
  barrier certificates.
\newblock {\em IEEE Transactions on Automatic Control}, 52(8):1415--1429, 2007.

\bibitem{Puri94}
A.~Puri and P.~Varaiya.
\newblock Decidability of hybrid systems with rectangular differential
  inclusion.
\newblock {\em {\upshape In} CAV'94, LNCS}, 3114:95--104, 1994.

\bibitem{SQ82}
J.-P. Queille and J.~Sifakis.
\newblock Specification and verification of concurrent systems in {CESAR}.
\newblock {\em {\upshape In} Proceedings of the 5th Colloquium on International
  Symposium on Programming}, pages 337--351, 1982.

\bibitem{Tiwari05}
E.~Rodr\'{i}guez-Carbonell and A.~Tiwari.
\newblock Generating polynomial invariants for hybrid systems.
\newblock In {\em HSCC 2005}, volume 3414 of {\em LNCS}, pages 590--605.
  Springer, 2005.

\bibitem{Sankar10}
S.~Sankaranarayanan.
\newblock Automatic invariant generation for hybrid systems using ideal fixed
  points.
\newblock {\em {\upshape In} HSCC'10, ACM}, pages 221--230, 2010.

\bibitem{Manna04}
S.~Sankaranarayanan, H.~Sipma, and Z.~Manna.
\newblock Constructing invariants for hybrid systems.
\newblock {\em {\upshape In} HSCC'04, LNCS}, 2993:539--554, 2004.

\bibitem{Tiwari10}
A.~Taly, S.~Gulwani, and A.~Tiwari.
\newblock Synthesizing switching logic using constraint solving.
\newblock {\em {\upshape In} VMCAI'09, LNCS}, 5403:305--319, 2009.

\bibitem{Taly09}
A.~Taly and A.~Tiwari.
\newblock Deductive verification of continuous dynamical systems.
\newblock {\em {\upshape In} FSTTCS09, LIPIcs}, 4:383--394, 2009.

\bibitem{tarski51}
A.~Tarski.
\newblock {\em A Decision Method for Elementary Algebra and Geometry}.
\newblock University of California Press, Berkeley, 1951.

\bibitem{Pollard}
M.~Tenenbaum and H.~Pollard.
\newblock {\em Ordinary Differential Equations}.
\newblock Dover Publications, Oct. 1985.

\bibitem{Xia07}
B.~Xia.
\newblock {DISCOVERER}: A tool for solving semi-algebraic systems.
\newblock {\em ACM SIGSAM Bulletin}, 41:102--103, 2007.

\bibitem{yzzx10}
L.~Yang, C.~Zhou, N.~Zhan, and B.~Xia.
\newblock Recent advances in program verification through computer algebra.
\newblock {\em Frontiers of Computer Science in China}, 4:1--16, 2010.

\bibitem{Zhang08}
S.~Zhang.
\newblock {\em The General Technical Solutions to Chinese Train Control System
  at Level 3 ({CTCS}-3)}.
\newblock China Railway Publisher, 2008.

\bibitem{Zhou95}
C.~Zhou, J.~Wang, and A.~P. Ravn.
\newblock A formal description of hybrid systems.
\newblock {\em {\upshape In} Hybrid Systems III, LNCS}, 1066:511--530, 1995.

\end{thebibliography}
%



\end{document}